\documentclass[a4paper,USenglish,cleveref, autoref, thm-restate,authorcolumns]{lipics-v2021}

\pdfoutput=1 
\hideLIPIcs  

\usepackage{graphicx}
\usepackage{wrapfig}
\usepackage{color}
\usepackage{comment}
\usepackage{url}
\usepackage{amsmath,amssymb,amsfonts,amsthm,mathtools}
\usepackage{algorithm}
\usepackage{algpseudocode}
\usepackage{orcidlink}

\usepackage{xcolor}
\hypersetup{
	colorlinks=true,
	citecolor=blue,
	linkcolor=blue,
}
\usepackage{apptools}
\newcommand{\restateref}[1]{\IfAppendix{\hyperref[#1]{$\clubsuit$}}{\hyperref[#1*]{$\clubsuit$}}}

\theoremstyle{plain}
\newtheorem{thm}{Theorem}

\newtheorem{cor}{Corollary}

\newcommand{\dptable}{\mathsf{dp}}

\usepackage[subrefformat=parens]{subcaption}
\usepackage{tikz}
\usetikzlibrary{calc}
\usetikzlibrary{arrows.meta}

\newcommand{\FV}{}  

\title{Finding Order-Preserving Subgraphs} 

\author{Haruya Imamura}{Kyushu Institute of Technology, Japan}{imamura.haruya389@mail.kyutech.jp}{https://orcid.org/0009-0007-0901-6542}{}

\author{Yasuaki Kobayashi}{Hokkaido University, Sapporo, Hokkaido, Japan}{koba@ist.hokudai.ac.jp}{https://orcid.org/0000-0003-3244-6915}{JSPS KAKENHI Grant Numbers 
JP23K28034, 
JP24H00686, 
JP24H00697.
}
\author{Yota Otachi}{Nagoya University, Nagoya, Aichi, Japan}{otachi@nagoya-u.jp}{https://orcid.org/0000-0002-0087-853X}{JSPS KAKENHI Grant Numbers
JP22H00513, 
JP24H00697, 
JP25K03076, 
JP25K03077. 
}
\author{Toshiki Saitoh}{Kyushu Institute of Technology, Iizuka, Fukuoka, Japan}{toshikis@ai.kyutech.ac.jp}{https://orcid.org/0000-0003-4676-5167}{JSPS KAKENHI Grant Numbers JP24H00697, JP24K14827.}
\author{Keita Sato}{Muroran Institute of Technology, Muroran, Hokkaido, Japan}{}{}{}
\author{Asahi Takaoka}{Muroran Institute of Technology, Muroran, Hokkaido, Japan}{takaoka@muroran-it.ac.jp}{https://orcid.org/0000-0002-0194-7138}{JSPS KAKENHI Grant Number JP23K03191.}
\author{Ryo Yoshinaka}{Tohoku University, Sendai, Miyagi, Japan}{ryoshinaka@tohoku.ac.jp}{https://orcid.org/0000-0002-5175-465X}{JSPS KAKENHI Grant Numbers
JP24H00697, 
JP24K14827. 
}
\author{Tom C. {van der Zanden}}{Maastricht University, Maastricht, The Netherlands}{t.vanderzanden@maastrichtuniversity.nl}{https://orcid.org/0000-0003-3080-3210}{}

\authorrunning{H. Imamura et al.} 

\Copyright{Haruya Imamura, Yasuaki Kobayashi, Yota Otachi, Toshiki Saitoh, Keita Sato, Asahi Takaoka, Ryo Yoshinaka, and Tom C. van der Zanden} 



\category{} 

\relatedversion{} 

\acknowledgements{We would like to thank Takeshi Tokuyama for providing valuable information. }

\nolinenumbers 

\EventEditors{John Q. Open and Joan R. Access}
\EventNoEds{2}
\EventLongTitle{42nd Conference on Very Important Topics (CVIT 2016)}
\EventShortTitle{CVIT 2016}
\EventAcronym{CVIT}
\EventYear{2016}
\EventDate{December 24--27, 2016}
\EventLocation{Little Whinging, United Kingdom}
\EventLogo{}
\SeriesVolume{42}
\ArticleNo{23}

\begin{document}

\maketitle

\begin{abstract} 
\textsc{(Induced) Subgraph Isomorphism} and \textsc{Maximum Common (Induced) Subgraph} are fundamental problems in graph pattern matching and similarity computation. In graphs derived from time-series data or protein structures, a natural total ordering of vertices often arises from their underlying structure, such as temporal sequences or amino acid sequences. This motivates the study of problem variants that respect this inherent ordering. This paper addresses \textsc{Ordered (Induced) Subgraph Isomorphism} (O(I)SI) and its generalization, \textsc{Maximum Common Ordered (Induced) Subgraph} (MCO(I)S), which seek to find subgraph isomorphisms that preserve the vertex orderings of two given ordered graphs.
Our main contributions are threefold:
(1) We prove that these problems remain NP-complete even when restricted to small graph classes, such as trees of depth 2 and threshold graphs.
(2) We establish a gap in computational complexity between OSI and OISI on certain graph classes. For instance, OSI is polynomial-time solvable for interval graphs with their interval orderings, whereas OISI remains NP-complete under the same setting.
(3) We demonstrate that the tractability of these problems can depend on the vertex ordering. For example, while OISI is NP-complete on threshold graphs, its generalization, MCOIS, can be solved in polynomial time if the specific vertex orderings that characterize the threshold graphs are provided.
\end{abstract}
%
%
\section{Introduction} 
 Given two graphs $G$ and $H$, \textsc{(Induced) Subgraph Isomorphism} asks whether $G$ contains a graph isomorphic to $H$ as an (induced) subgraph. 
 \textsc{Maximum Common (Induced) Subgraph}, which is a generalization of \textsc{(Induced) Subgraph Isomorphism}, is a problem to find, given two graphs $G$ and $H$, a graph $Z$ with the maximum number of edges (vertices) such that $Z$ is an (induced) subgraph of both $G$ and $H$. 
 These problems are NP-hard because they generalize many other fundamental and important NP-complete problems such as \textsc{Hamiltonian Path}, \textsc{Independent Set}, and \textsc{Clique}. 

For these problems, there are some tractability and intractability results for several cases where input graphs belong to subclasses of perfect graphs. 
First, if the smaller graph $H$ is disconnected, the problems are NP-complete even for disjoint unions of paths~\cite{Damaschke90,GJ79}.
Hence, when considering graph classes that include disjoint unions of paths, we particularly focus on the case where $H$ is connected.
Kijima et al.~\cite{KijimaOSU12} showed the NP-completeness of \textsc{Subgraph Isomorphism} on proper interval graphs, bipartite permutation graphs, and trivially perfect graphs. 
They also proposed polynomial-time algorithms for chain graphs, cochain graphs, and threshold graphs.
Konagaya et al.~\cite{KonagayaOU16} studied the complexity of the problem when $H$ belongs to a small class of graphs.
For \textsc{Induced Subgraph Isomorphism}, the problem is NP-complete on interval graphs~\cite{MarxS13} and cographs~\cite{Damaschke90}, while it can be solved in polynomial time for proper interval graphs, bipartite permutation graphs~\cite{HeggernesHMV15}, and threshold graphs~\cite{BelmonteHH12}. 

Many graph classes discussed above are intersection graphs of some geometric objects.
For example, an interval graph is the intersection graph of a family of intervals on the real line.
We can naturally define a vertex ordering (i.e., a total order on the vertex set) from the set of intervals. 
In this setting, isomorphism mappings are expected to be order-preserving; that is, a vertex $u$ precedes a vertex $v$ in $H$ if and only if $f(u)$ precedes $f(v)$ in $G$ for a subgraph isomorphism $f$ from $H$ to $G$.
Such a mapping is particularly relevant in pattern-matching problems on graphs where each vertex corresponds to an event in a time series.
As another application, proteins can be represented as sequences of amino acids. A graph constructed based on the distances between the carbon atom of each amino acid is called a protein contact network, and the resulting graph is an intersection graph of unit balls with an inherent vertex ordering~\cite{di2013protein}.
In these applications, we aim to perform pattern matching or compute similarity on graphs with vertex orderings.
Therefore, we address the problems \textsc{Ordered (Induced) Subgraph Isomorphism} (O(I)SI) and \textsc{Maximum Common Ordered  (Induced) Subgraph} (MCO(I)S): 
Given two graphs along with their vertex orderings, we aim to find an order-preserving subgraph isomorphism or an order-preserving maximum common subgraph (see Section 2 for the formal definition). 

Here, we discuss the computational complexity of the problems.
The problems O(I)SI are generally NP-complete because they include \textsc{Clique} due to the vertex symmetry in complete graphs.
On the other hand, when the two input graphs have the same number of vertices, both OISI and OSI can be trivially solved by providing vertex orderings.
Note that in this case of non-ordered problems, \textsc{Induced Subgraph Isomorphism} is equivalent to the graph isomorphism problem, whereas \textsc{Subgraph Isomorphism} includes \textsc{Hamiltonian Cycle} and is therefore NP-complete.
Bose et al.~\cite{BoseBL98} showed that the \textsc{Order Preserving Subsequence} (OPS) problem for two integer sequences is NP-complete. 
This problem is equivalent to OISI on permutation graphs with naturally defined vertex orderings of permutation graphs. 
In contrast, they also proposed a polynomial-time algorithm for cographs with the same vertex orderings.
Guillemot and Marx~\cite{GuillemotM14} presented a linear-time algorithm for OISI on permutation graphs with the orderings when the number of vertices in $H$ is constant. 
Heggernes et al.~\cite{HeggernesHMV15} provided a polynomial-time algorithm for OSI on interval graphs with interval orderings that characterize interval graphs as an exercise.
In the preliminary version of~\cite{HeggernesHMV15}, they~\cite{HeggernesMV10} also provided a polynomial-time algorithm for OISI on interval graphs with interval orderings, 
but they have been pointed out that it contains a bug in~\cite{HeggernesHMV15}. 
These studies have considered only the natural orderings that characterize the graph classes, without addressing arbitrary orderings.


\begin{figure}[tb]
    \centering
    \includegraphics[width = \textwidth]{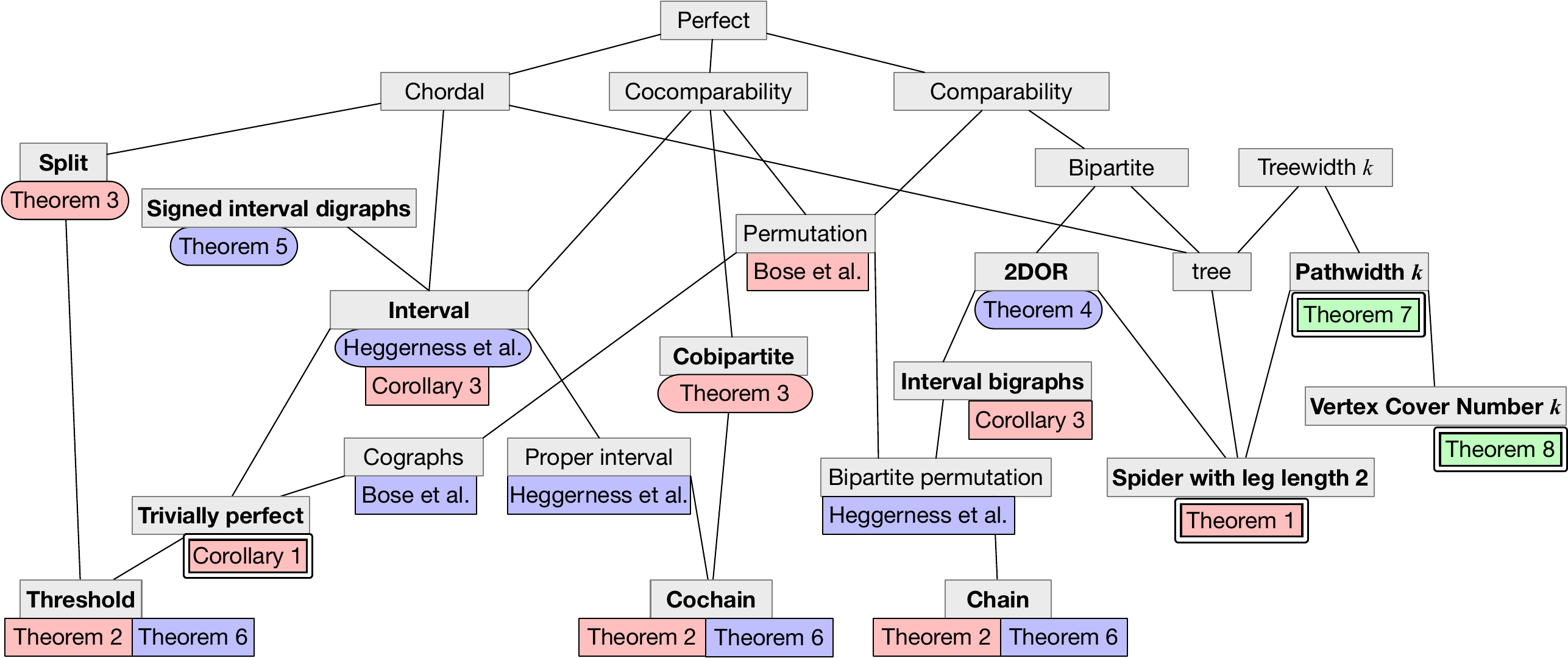}
    \caption{Illustration of the inclusion relations among graph classes along with the (in)tractability of OSI and OISI. Red, blue, and green denote NP-completeness, polynomial-time solvability, and fixed-parameter tractability, respectively. OSI, OISI, and both problems are represented by squashed rectangles, single-stroked rectangles, and double-stroked rectangles, respectively.}
    \label{fig:graphclasses}
\end{figure}
\noindent \textbf{Our contributions: }
In this paper, we discuss the O(I)SI problems and their generalization, MCO(I)S, with a focus not only on graph classes but also on the properties of orderings. 
The contributions of this paper are as follows, and they are illustrated in \figurename~\ref{fig:graphclasses}.

 We prove in \cref{sec:hard-tree} that both OSI and OISI are NP-complete, even on trees. The trees obtained in our reductions are spiders with leg length~2, which are two-directional orthogonal ray graphs (or 2DOR graphs) with pathwidth~2. 
 By a similar reduction, we can also show that both problems are NP-complete for trivially perfect graphs, a subclass of interval graphs.
 Furthermore, for OISI, we show NP-completeness for three smaller graph classes: threshold graphs, chain graphs, and cochain graphs in \cref{sec:hard-threshold}. 
 The result on threshold graphs (\cref{cor:interval}) implies that OISI remains NP-complete on interval graphs even if the vertex orderings are restricted to interval orderings, which indicates that a bug in the algorithm of~\cite{HeggernesMV10} is unlikely to be fixable.
 We also prove that OSI remains NP-complete even on split graphs with perfect elimination orderings and cobipartite graphs with cocomparability orderings in Section~\ref{sec:hard-split}. 
 

 For OSI, we propose polynomial-time algorithms for two graph classes related to interval graphs and interval bigraphs in Section~\ref{sec:alg-OSI}, provided that the vertex orderings are ``good'': One is the class of 2DOR graphs with comparability weak orderings, and the other is the class of signed-interval digraphs with min orderings.
 In Section~\ref{sec:hard-tree}, we show that for these graph classes, OSI is NP-complete (without restrictions on vertex orderings).
 This indicates that there is a difference in the complexity between OSI and OISI: For these graph classes, OSI can be solved in polynomial time using a ``good'' ordering, whereas OISI remains NP-complete even with such orderings~(\cref{cor:interval}).

 We propose algorithms for MCO(I)S, a generalization of O(I)SI.
 For MCOIS, we propose polynomial-time algorithms for threshold graphs, chain graphs, and cochain graphs with inclusion orderings in Section~\ref{sec:alg-threshold}. These graph classes can be characterized by inclusion orderings, and our algorithms are based on dynamic programming on these orderings. 
 For MCO(I)S, we propose fixed-parameter algorithms (FPT algorithms), parameterized by the pathwidth of the input vertex ordering in Section~\ref{sec:FPTalg-pathwidth}. Here, the pathwidth of a vertex ordering is defined based on an alternative characterization of pathwidth, known as the vertex separation number~\cite{Kinnersley92}.
 For these graphs, the problems are NP-complete on arbitrary orderings in Section~\ref{sec:hard-tree} and~\ref{sec:hard-threshold}. 
 These results suggest that the vertex orderings have a significant impact on the complexity of our problems. 
 Finally, we present FPT algorithms for MCO(I)S, parameterized by the vertex cover number of the input graph in Section~\ref{sec:FPTalg-VC}. 
 We note that the algorithm works on arbitrary vertex orderings and that the pathwidth of the input vertex ordering can be arbitrarily large even when the vertex cover number is bounded. 

\section{Definitions}
For a natural number $n$, we denote the set of integers $\{1, 2, \dots, n\}$ by $[n]$. 

Unless stated otherwise, graphs are assumed to be simple and undirected. 
A graph $G$ is an ordered pair $(V_G , E_G)$, where $V_G$ is the vertex set of $G$ and $E_G$ is the edge set of $G$. 
Let $n_G$ be the number of vertices of $G$.
The neighborhood of a vertex $v$ is the set $N(v) = \{u \in V_G \mid \{u, v\} \in E_G\}$, and the degree of $v$ is $|N(v)|$ denoted by $d(v)$.
If $d(v) = 1$, $v$ is a \emph{leaf}.
A sequence of pairwise distinct vertices $P = (v_1, v_2, \dots, v_k)$ is a \emph{path} from $v_1$ to $v_k$ of length $k-1$ in $G$ provided that $\{v_i, v_{i+1}\}\in E_G$ for $i \in [k-1]$.
A graph $G$ is connected if for any two vertices $u, v$, there exists a path from $u$ to $v$.
For a subset $V' \subseteq V_G$, the subgraph of $G$ induced by $V'$ is denoted as $G[V']$.
If every pair of vertices in $V'$ is adjacent (resp.\ non-adjacent) in $G$, $V'$ is called a \emph{clique} (resp.\ an \emph{independent set}).
A graph $G$ is \emph{bipartite} if the vertex set of $G$ can be partitioned into two independent sets. 
A graph $\overline{G} = (V_G, \overline{E_G})$ is called the \emph{complement} of $G$, where $\overline{E_G} = \{\{u, v\} \mid \{u, v\}\notin E_G\}$.
A graph $G$ is \emph{cobipartite} if $\overline{G}$ is bipartite. 

A vertex ordering of $G$ is a linear ordering $\sigma = (v_1, ... ,v_n)$ of the vertices of $G$.
We write $v_i \prec_\sigma v_j$ if $i < j$. 
A graph $H = (V_H, E_H)$ is \emph{subgraph-isomorphic} to a graph $G = (V_G, E_G)$ if there exists an injective map $f$ from $V_H$ to $V_G$ such that $\{u, v\} \in E_H$ implies $\{f(u), f(v)\} \in E_G$.
The map $f$ is called a \emph{subgraph isomorphism}. 
If the map $f$ satisfies $\{u, v\} \in E_H \Leftrightarrow \{f(u), f(v)\} \in E_G$, then $f$ is an \emph{induced subgraph isomorphism} and $H$ is \emph{induced subgraph-isomorphic} to $G$. 
Let $\sigma$ and $\tau$ be vertex orderings of $G$ and $H$, respectively.
An injective map $f\colon V_H \to V_G$ is $(\sigma, \tau)$-\emph{injective}, if $f(u) \prec_\sigma f(v)$, for every pair $u,v$ of vertices in $H$ with $u \prec_\tau v$.
For two graphs $G$ and $H$, if a $(\sigma, \tau)$-injective map $f$ is a (induced) subgraph isomorphism from $H$ to $G$, $f$ is an \emph{ordered (induced) subgraph isomorphism}. 
If there is an ordered (induced) subgraph isomorphism from $H$ to $G$, then $H$ is \emph{ordered (induced) subgraph-isomorphic} to $G$.
The problems addressed in this paper are defined as follows.

\medskip
\noindent \textsc{Ordered (Induced) Subgraph Isomorphism} (O(I)SI)\\
\textbf{Input}: Graphs $G=(V_G, E_G)$ and $H=(V_H, E_H)$ with vertex orderings $\sigma$ and $\tau$, respectively, where $|V_G| \ge |V_H|, |E_G| \ge |E_H|$.\\
\textbf{Question}: Is $H$ ordered (induced) subgraph-isomorphic to $G$?
\medskip

\noindent \textsc{Maximum Common Ordered (Induced) Subgraph} (MCO(I)S)\\
 \textbf{Input}: Graphs $G=(V_G, E_G)$ and $H=(V_H, E_H)$ with vertex orderings $\sigma$ and $\tau$, respectively, and a positive integer $k$.\\
 \textbf{Question}: Is there a graph $Z$ with $|E_Z|\geq k$ ($|V_Z|\geq k$) and its vertex ordering such that $Z$ is ordered (induced) subgraph-isomorphic to both $G$ and $H$. 
\medskip

Let a triple $(Z, f_G, f_H)$ be a solution of $G$ and $H$ for MCO(I)S, where $f_G$ and $f_H$ are ordered subgraph isomorphisms from $Z$ to $G$ and $H$, respectively. 
A vertex $v\in V_G$ is \emph{matched} with $u\in V_H$ if $f^{-1}_G(v) = f^{-1}_H(u)$.

\section{NP-completeness}
\subsection{O(I)SI on Spiders with Leg Length 2}
\label{sec:hard-tree}
We show that both OSI and OISI are NP-complete even for spiders with leg length 2. 
A tree is called a \emph{spider} if it has just one vertex of degree 3 or more, which is called the \emph{central vertex}.
Each path from the central vertex to a leaf is called a \emph{leg}.
Our reduction is from \textsc{Order Preserving Subsequence} (OPS), which asks, given two permutations $\pi$ over $[n]$ and $\rho$ over $[k]$ with $n \ge k$, whether there exists an index subset $\{i_1, i_2, \dots, i_k\}\subseteq [n]$ with $i_1 < i_2 < \dots < i_k$ such that for any $j_1, j_2\in [k]$, $\pi(i_{j_1}) < \pi(i_{j_2})$ if and only if $\rho(j_1) < \rho(j_2)$. 
This problem is known to be NP-complete~\cite{BoseBL98}.

\begin{thm}\label{thm:NPcomp_tree}
 OSI and OISI are NP-complete for spiders with leg length 2. 
\end{thm}
\begin{proof}
 We construct two spiders $T_G$ and $T_H$ and their vertex orderings from the instance $\pi$ over $[n]$ and $\rho$ over $[k]$ of OPS.
 The vertex set of the spider $T_G$ is $V_G = \{v_0, v_1\dots,v_{2n}\}$ with ordering $\sigma = (v_0, v_1, \dots, v_{2n})$.
 The central vertex is $v_0$ and the legs are $(v_0,v_i,v_{n+\pi(i)})$ for $i \in [n]$. 
 Similarly, the spider $T_H$ has vertices $V_H = \{u_0, u_1\dots,u_{2k}\}$ with ordering $\tau = (u_0, u_1, \dots, u_{2k})$ and legs $(u_0,u_i,u_{k+\rho(i)})$ for $i \in [k]$. 
 We depict the spiders in \figurename~\ref{fig:reduce-tree}. 

 \begin{figure}[tb]
  \centering
  \includegraphics[width=.6\textwidth]{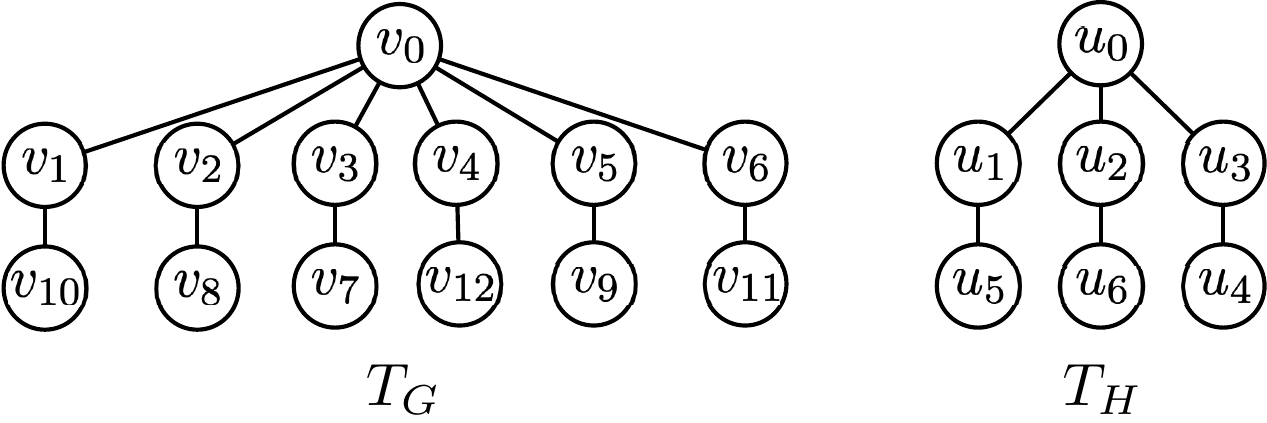}
  \caption{The reduced instance of O(I)SI from $\pi = (4, 2, 1, 6, 3, 5)$ and $\rho = (2, 3, 1)$ of \textsc{Order Preserving Subsequence}. The answer of \textsc{Order Preserving Subsequence} is $(4, 6, 3)$.}
  \label{fig:reduce-tree}
 \end{figure}

 We show that the instance $(\pi,\rho)$ of OPS is a ``yes'' instance if and only if $T_H$ is ordered subgraph-isomorphic to $T_G$.
 In the following, we only consider the non-induced case as every ordered subgraph isomorphism from $T_H$ to $T_G$ is an ordered induced subgraph isomorphism.

 Assume indices $i_1,\dots,i_k \in [n]$ with $i_1 < i_2 < \dots < i_k$ witness that $(\pi,\rho)$ is a ``yes'' instance of OPS:
 for any $j_1, j_2\in [k]$, $\pi(i_{j_1}) < \pi(i_{j_2})$ if and only if $\rho(j_1) < \rho(j_2)$. 
 We define a mapping $f\colon V_{H} \to V_{G}$ by $f(u_0) = v_0$ and $f(u_j) = v_{i_j}$ and $f(u_{k+\rho(j)}) = v_{n+\pi(i_j)}$ for $j\in [k]$. 
 Obviously, $f$ is a subgraph isomorphism from $T_H$ to $T_G$. 
 It remains to show the order consistency.
 Suppose $u_{j_1} \prec_\tau u_{j_2}$, i.e., $j_1 < j_2$.
 The only nontrivial case is when $j_1,j_2 > k$.
Let $j_1' = \rho^{-1}(j_1-k)$ and $j_2' = \rho^{-1}(j_2-k)$.
Then,
\begin{multline*}
  u_{j_1} \prec_\tau u_{j_2}
  \iff
  j_1 = k+\rho(j_1') < k+\rho(j_2') = j_2
  \iff
  \rho(j_1') < \rho(j_2')
\\  \iff
  \pi(i_{j_1'}) < \pi(i_{j_2'})
  \iff 
  f(u_{j_1}) = v_{n+\pi(i_{j_1'})} \prec_\sigma v_{n+\pi(i_{j_2'})} = f(u_{j_2})
  \,.
\end{multline*}

 Conversely, suppose that there exists an ordered subgraph isomorphism $f$ from $T_H$ to $T_G$. 
 Due to the degrees of $v_0$ and $u_0$, it follows that $f(u_0) = v_0$. 
 The neighbors $u_1, \dots, u_k$ of $u_0$ in $T_H$ must be mapped to some of the neighbors $v_1, \dots, v_n$ of $v_0$ in $T_G$.
 Let $i_j \in [n]$ be such that $f(u_j) = v_{i_j}$ for $j\in [k]$. 
 We show that those indices $i_1,\dots,i_k$ provide a witness for $(\pi,\rho)$ to be a ``yes'' instance of OPS.
 %
 If $j_1 < j_2$ for $j_1, j_2\in [k]$, $i_{j_1} < i_{j_2}$ follows from the orderings of $\sigma$ and $\tau$. 
 Additionally, $f(u_{k+\rho(j)}) = v_{n+\pi(i_{j})}$ since $u_{k+\rho(j)}$ and $v_{n+\pi(i_j)}$ are neighbors of $u_j$ and $v_{i_j}$ in $T_H$ and $T_G$, respectively. 
 Since $f$ is an ordered (induced) subgraph isomorphism, $v_{n+\pi(i_{j_1})}\prec_{\sigma} v_{n+\pi(i_{j_2})}$ if and only if $u_{k+\rho(j_1)}\prec_{\tau} u_{k+\rho(j_2)}$. 
 Therefore, it follows that $\pi(i_{j_1}) < \pi(i_{j_2})$ if and only if $\rho(j_1) < \rho(j_2)$ for any $j_1, j_2\in [k]$. 
\end{proof}
The graphs used in the proof of Theorem~\ref{thm:NPcomp_tree} are actually 2DOR graphs, defined in \cref{sec:alg-OSI}, with pathwidth 2.
This contrasts with the tractability results of Theorems~\ref{thm:algo 2DOR graphs} and~\ref{thm:FPT_pathdecomposition},
which demonstrate that if the given orderings are ``inherent'' to the graph structures, the problem may become easier to solve.

A graph is \emph{trivially perfect} if it contains neither an induced path of length four nor an induced cycle of length four~\cite{GraphClasses}.
A trivially perfect graph can be obtained from the graph used in the proof of \cref{thm:NPcomp_tree} by adding edges between the center vertex and the leaves. 
The same argument as in the proof applies to these graphs.
\begin{cor}
 OSI and OISI are NP-complete even on trivially perfect graphs. 
\end{cor}

Although this paper is focused on connected graphs, we present here, as an exception, a result concerning disconnected graphs.
We used the central vertices of the spiders in the proof of Theorem~\ref{thm:NPcomp_tree} to make the graphs connected. 
The reduction still works if we remove the central vertices.
\begin{cor}
 OSI and OISI are NP-complete even for graphs consisting of disjoint edges. 
\end{cor}


\subsection{OISI on Threshold/Chain/Cochain Graphs}
\label{sec:hard-threshold}
Let $G=(V_G, E_G)$ be a graph with a bipartition $(X, Y)$ of $V_G$. 
For a vertex $x\in X$, $N_Y(x)$ denotes the neighbor set of $x$ in $Y$, that is, $N_Y(x) = N(x)\cap Y$. 
For a vertex set $X'\subseteq X$, $N_Y(X') = \bigcup_{x\in X'} N_Y(x)$. 
The neighbor sets $N_X(y)$ for $y\in Y$ and $N_X(Y')$ for $Y'\subseteq Y$ are defined similarly.
An ordering $(x_1, x_2, \dots, x_{|X|})$ on $X$ is an \emph{inclusion ordering} if $N_Y(x_i) \subseteq N_Y(x_j)$ for every $i, j$ with $i < j$. 
Note that if there exists an inclusion ordering on $X$, then there also exists an inclusion ordering $(y_1, y_2, \dots, y_{|Y|})$ on $Y$ such that if $N_X(y_i) \subseteq N_X(y_j)$ for every $i, j$ with $i < j$. 
Suppose that $X$ (and hence $Y$) has an inclusion ordering in $G$. 
We say that $G$ is \emph{threshold} if $X$ is a clique and $Y$ is an independent set, 
$G$ is \emph{chain} if $G$ is bipartite, and $G$ is \emph{cochain} if $G$ is cobipartite. 
An \emph{inclusion ordering} of $G$ is defined as the inclusion ordering on $Y$ followed by the inclusion ordering on $X$.
\ifdefined\APP
\begin{restatable}[\restateref{thm:NPComp-Threshold}]{thm}{thmNPCompThreshold}\label{thm:NPComp-Threshold}
 OISI is NP-complete even for threshold graphs, chain graphs, and cochain graphs. 
\end{restatable}
\else
\begin{thm}\label{thm:NPComp-Threshold}
 OISI is NP-complete even for threshold graphs, chain graphs, and cochain graphs. 
\end{thm}
\fi

\ifdefined\FV
\ifdefined\APP
\thmNPCompThreshold*\label{thm:NPComp-Threshold*}
\fi

\begin{proof}
 \begin{figure}[tb]
  \centering
  \includegraphics[width=.7\textwidth]{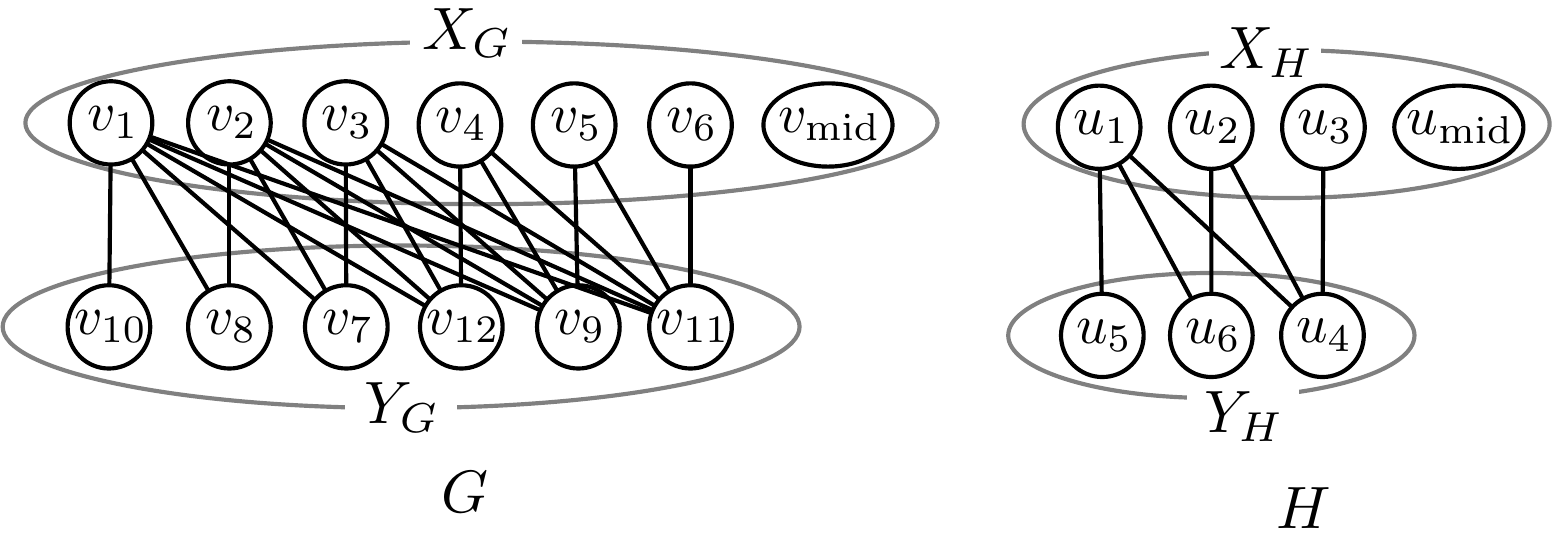}
  \caption{The reduced instance from $\pi = (4, 2, 1, 6, 3, 5)$ and $\rho = (2, 3, 1)$.}
  \label{fig:reduce-threshold}
 \end{figure}
 Let a pair $(\pi, \rho)$ be an instance of \textsc{Order Preserving Subsequence}, where $|\pi|=n$ and $|\rho|=k$. 
 The graph $G$ consists of $2n+1$ vertices, $V(G) = \{v_1, \dots, v_{2n}, v_{\mathrm{mid}}\}$, and its vertex ordering is defined as $\sigma = (v_1, \dots, v_n, v_{\mathrm{mid}}, v_{n+1}, \dots, v_{2n})$. 
 We partition the vertex set as $X_G = \{v_1, \dots, v_n, v_{\mathrm{mid}}\}$ and $Y_G = \{v_{n+1}, \dots, v_{2n}\}$. 
 For each vertex $v_i$ with $i\in [n]$, the neighbor set satisfies $N(v_i)\cap Y_G = \{v_{n+\sigma(i)}, v_{n+\sigma(i+1)}, \dots, v_{n+\sigma(n)}\}$, 
 and for the vertex $v_{\mathrm{mid}}$, we have $N(v_{\mathrm{mid}}) \cap Y_G = \emptyset$. 
 Note that the ordering $(v_{\mathrm{mid}}, v_n, \dots, v_1)$ is an inclusion ordering, that is, depending on whether $X_G$ and $Y_G$ are cliques or independent sets, the graph becomes a threshold graph, a chain graph, or a cochain graph.
 Similarly, the vertex set of graph $H$ is defined as $V(H) = \{u_1, \dots, u_{2k}, u_{\mathrm{mid}}\}$, with the vertex ordering given by $\tau = (u_1, \dots, u_k, u_{\mathrm{mid}}, u_{k+1}, \dots, u_{2k})$. 
 We partition the vertex set into $X_H = \{u_1, \dots, u_k, u_{\mathrm{mid}}\}$ and $Y_H = \{u_{k+1}, \dots, u_{2k}\}$. 
 For each vertex $u_i$ with $i\in [k]$, the neighbor set is defined by $N(u_i)\cap Y_H = \{u_{k+\sigma(i)}, u_{k+\sigma(i+1)}, \dots, u_{k+\sigma(k)}\}$, and for the vertex $u_{\mathrm{mid}}$, $N(u_{\mathrm{mid}}) \cap Y_H = \emptyset$. 
 See an example in \figurename~\ref{fig:reduce-threshold}. 

 Using the same discussion of the proof of Theorem~\ref{thm:NPcomp_tree}, we can show the following. 
 For a witness $(\pi(i_1), \pi(i_2), \dots, \pi(i_k))$ of \textsc{Order Preserving Subsequence}, we obtain an ordered subgraph isomorphism $f$: $f(u_{\mathrm{mid}}) = v_{\mathrm{mid}}$, and for $j\in [k], f(u_j) = v_{i_j}$ and $f(u_{k+\rho(j)}) = v_{n+\pi(i_j)}$. 
 On the other hand, we can construct a subsequence $(\pi(i_1), \pi(i_2), \dots, \pi(i_k))$ as a witness of \textsc{Order Preserving Subsequence} from an ordered induced subgraph isomorphism $f$ as follows: 
 for each $j\in [k]$, $f(u_{k+\rho(j)}) = v_{n+\sigma(j')}$ for $j'=\sigma(\pi(i_j))$. 

 We note that chain graphs are not connected because $u_{\mathrm{mid}}$ and $v_{\mathrm{mid}}$ are isolated vertices. 
 However, it does not need $v_{\mathrm{mid}}$ and $u_{\mathrm{mid}}$ since the bipartition is unique in a connected bipartite graph. 
 This ensures that $|\{f(u_{k+1}), \dots, f(u_{2k})\}\cap Y_G| = k$. 
 Therefore, we can also state NP-hardness for connected chain graphs. 
\end{proof}

\fi


\noindent
\textbf{Interval orderings and interval bigraph orderings:}
We here show the NP-completeness for OISI on interval graphs with interval orderings and interval bigraphs with interval bigraph orderings using the graphs described in the proof of Theorem~\ref{thm:NPComp-Threshold}. 
Let $G = (V_G, E_G)$ be a graph with $V_G$.
Then, $G$ is an \emph{interval graph} if there exists a set of intervals $\mathcal{I} = \{I_v \mid v\in V_G\}$ such that $I_u\cap I_v \neq \emptyset$ for $u, v\in V_G$ if and only if $\{u, v\}\in E_G$. 
A vertex ordering $\sigma$ of $G$ is an \emph{interval ordering} if for every triple $u, v, w$ in $V(G)$ with $u\prec_{\sigma} v$ and $v\prec_{\sigma} w$, $\{u, w\}\in E_G$ implies $\{v, w\}\in E_G$.
A graph $G$ is an interval graph if and only if it admits an interval ordering~\cite{Wood06,HeggernesMV10}. 
Note that the class of interval graphs is a superclass of threshold graphs~\cite{golumbic04}. 

A bipartite graph $G=(X_G\cup Y_G, E_G)$ is an \emph{interval bigraph} if there exists a set of intervals $\mathcal{I} = \{I_v \mid v\in X\cup Y\}$ such that $I_x\cap I_y \neq \emptyset$ for $x\in X_G, y\in Y_G$ if and only if $\{x, y\}\in E_G$. 
A vertex ordering $\sigma$ of $G$ is an \emph{interval bigraph ordering} if for every triple $u, v\in X_G$ and $w\in Y_G$ or $u, v\in Y_G$ and $w\in X_G$ with $u\prec_{\sigma} v$ and $v\prec_{\sigma} w$, and $\{u, w\}\in E_G$ implies $\{v, w\}\in E_G$.
A graph $G$ is an interval bigraph if and only if it admits an interval bigraph ordering~\cite{HellH04}. 
Note that the class of interval bigraphs is a superclass of chain graphs. 

\ifdefined\APP
\begin{restatable}[\restateref{cor:interval}]{cor}{corInterval}\label{cor:interval}
 OISI is NP-complete even on threshold graphs with interval orderings, and chain graphs with interval bigraph orderings. 
\end{restatable}
\else
\begin{cor}\label{cor:interval}
 OISI is NP-complete even on threshold graphs with interval orderings, and chain graphs with interval bigraph orderings. 
\end{cor}
\fi

\ifdefined\FV
\ifdefined\APP
\corInterval*\label{cor:interval*}
\fi

\begin{proof}
 To show the NP-completeness for interval graphs, we construct complement graphs $\overline{G}$ and $\overline{H}$ of the threshold graphs $G$ and $H$ described in the proof of Theorem~\ref{thm:NPComp-Threshold}. 
 Both graphs $\overline{G}$ and $\overline{H}$ are threshold graphs, that is, they are interval graphs. 
 It is easy to check that the vertex orderings $\sigma$ and $\tau$ are both interval orderings. 
 The answer of the OISI on $\overline{G}$ and $\overline{H}$ is the same as $G$ and $H$ since every induced subgraph isomorphism from $H$ to $G$ is also an induced subgraph isomorphism from $\overline{H}$ to $\overline{G}$ and vice versa.

 For interval bigraphs with interval bigraph orderings, we use chain graphs $G$ and $H$ described in the proof. 
 We define the vertex orderings $\sigma = (v_n, v_{n-1}, \dots, v_1, v_{\mathrm{mid}}, v_{2n}, \dots, v_{n+1})$ and $\tau=(u_{k}, \dots, u_{1}, u_{\mathrm{mid}}, u_{2k}, \dots, u_{k+1})$. 
 The orderings are interval bigraph orderings, and the remainder of the proof is identical to the one in \cref{thm:NPComp-Threshold}. 
\end{proof}

\fi

\subsection{OSI on Split Graphs and Cobipartite Graphs}
\label{sec:hard-split}
An edge that joins two non-consecutive vertices in a cycle is called a chord.
A graph $G$ is \emph{chordal} if every cycle with length at least 4 has a chord in $G$.
A vertex ordering $\sigma = (v_1, \dots, v_n)$ is a \emph{perfect elimination ordering} of $G$ if $N(v_i)\cap \{v_{i+1}, \dots, v_n\}$ is a clique for all $i\in [n]$.
A graph $G$ is a chordal graph if and only if there is a perfect elimination ordering of $G$ \cite{GraphClasses}.
A graph $G$ is \emph{split} if the vertex set of $G$ can be partitioned into an independent set $I$ and a clique $K$.
A graph $G$ is a split graph if and only if $G$ and the complement of $G$ are chordal~\cite{GraphClasses}.

A graph $G$ is a \emph{cocomparability graph} if the complement $\overline{G}$ of $G$ has a comparability ordering (see \cref{sec:alg-OSI} for the definition).
A vertex ordering $\sigma = (v_1, \dots, v_n)$ is a \emph{cocomparability ordering} if for every triple $i, j, k\in [n]$ with $i<j<k$, $\{v_i, v_k\}\in E_G$ implies $\{v_i, v_j\}\in E_G$ or $\{v_j, v_k\}\in E_G$.
A graph $G$ is cocomparability if and only if it has a cocomparability ordering~\cite{Wood06}.


\ifdefined\APP
\begin{restatable}[\restateref{thm-NPcomp-split}]{thm}{thmNPCompSplit}\label{thm-NPcomp-split}
    OSI is NP-complete even on split graphs with perfect elimination orderings and cobipartite graphs with cocomparability orderings.
\end{restatable}
\else
\begin{thm}\label{thm-NPcomp-split}
    OSI is NP-complete even on split graphs with perfect elimination orderings and cobipartite graphs with cocomparability orderings.
\end{thm}
\fi

\ifdefined\FV
\ifdefined\APP
\thmNPCompSplit*\label{thm-NPcomp-split*}
\fi

\begin{proof}
We give reductions from \textsc{Balanced Biclique}~\cite{GJ79}: Given a bipartite graph $G = (X_G\cup Y_G, E_G)$ and an integer $k$, the goal is to decide whether there is a biclique $A\cup B$ in $G$ such that $A\subseteq X_G, B\subseteq Y_G$, and $|A| = |B| = k$.
Since the vertices in $A$ and $B$ are symmetric, the orderings within $A$ and $B$ are negligible. 
We exploit this property in our reductions. 

\noindent
{\bf Split graphs with perfect elimination orderings: } Let $(G, k)$ with bipartite graph $G = (X_G\cup Y_G, E_G)$ and integer $k$ be an instance of \textsc{Balanced Biclique}.
    From the graph~$G$, we construct a split graph $G'$ as follows.
    We add a vertex $u_G$ and edges between every pair of vertices in $X_G \cup \{u_G\}$, forming a clique.
    Moreover, we add $n_G$ vertices that are adjacent to $u_G$ only, and the set of these vertices is denoted by $L_G$.
    The graph $G'$ is a split graph with a clique $X_G \cup \{u_G\}$ and an independent set $Y_G \cup L_G$.
    The vertex ordering $\sigma$ is defined as $(Y_G, L_G, u_G, X_G)$, where the internal orderings within $Y_G$, $L_G$, and $X_G$ are determined arbitrarily.
    This ordering is a perfect elimination ordering of $G'$ as the neighborhood of each vertex in $Y_G\cup L_G$ is a clique and the remaining $u_G\cup X_G$ is indeed a clique in $G'$.
    We construct another split graph $H'$ from a complete bipartite graph $H \coloneqq K_{k, k}$ with vertex set $X_H \cup Y_H$ and a perfect elimination ordering $\tau = (Y_H, L_H, u_H, X_H)$ in a similar manner.

    When $H'$ is subgraph-isomorphic to $G'$, the vertex $u_H$ must be mapped to $u_G$ based on its degree, and each vertex in $L_H$ must be mapped to a vertex in $L_G$.
    This implies that each vertex in $Y_H$ and in $X_H$ must be mapped to a vertex in $Y_G$ and in $X_G$, respectively due to the orderings $\sigma$ and $\tau$.
    Since vertices in $X_H$ (and those in $Y_H$) are indistinguishable, we can conclude that there is an ordered subgraph isomorphism from $H'$ to $G'$ if and only if $G$ has a balanced biclique $K_{k, k}$.

\noindent
{\bf Cobipartite graphs with cocomparability orderings: }
    The reduction is similar to the previous one.
    Let $(G, k)$ be the instance of \textsc{Balanced Biclique}.
    Starging from $G = (X_G \cup Y_G, E_G)$, we add two sets of vertices $X'_G$ and $Y'_G$ with $|X'_G| = |Y'_G| = n_G$ such that both $X_G \cup X'_G$ and $Y_G \cup Y'_G$ induce cliques.
    Moreover, we add a vertex $u_G$ that is adjacent to all the vertices.
    We define a vertex ordering $\sigma = (Y'_G, Y_G, u_G, X_G, X'_G)$, where the innner orderings within $Y'_G$, $Y_G$, $X_G$, and $X'_G$ are determined arbitrarily.
    Observe that the graph $G'$ and its ordering $\sigma$ are a cocomparability graph and a cocomparability ordering, respectively.
    We construct another cocomparability graph $H'$ from a complete bipartite graph $H \coloneqq K_{k, k}$ with vertex set $X_H \cup Y_H$ and a cocomparability ordering $\tau = (Y'_H, Y_H, u_H, X_H, X'_H)$ in a similar manner.

    By the same argument as used for split graphs, the vertex $u_H$ must be mapped to $u_G$ based on its degree, each vertex in $X_H$ (resp.~$Y_H$) must be mapped to $X_G$ (resp.~$Y_G$), and each vertex in $X'_H$ (resp.~$Y'_H$) must be mapped to $X'_G$ (resp.~$Y'_G$).
    Due to the intistinghishability of vertices, we can conclude that there is an ordered subgraph isomorphism from $H'$ to $G'$ if and only if $G$ has a balanced biclique $K_{k, k}$ as well.
\end{proof}

\fi

\section{Algorithms}
\subsection{OSI on 2DOR Graphs and Signed-interval Digraphs}
\label{sec:alg-OSI}
As stated in the introduction,
Heggernes et~al. \cite{HeggernesHMV15}
implicitly presented a polynomial-time algorithm 
for \textsc{OSI} 
on interval graphs with interval orderings. 
Building on this algorithm, 
we present polynomial-time algorithms for 
two-directional orthogonal ray graphs (or 2DOR graphs) 
and signed-interval digraphs in this section.
These three graph classes share a common ordering characterization, known as min orderings \cite{HHMR20-SIDMA}, 
and 2DOR graphs and signed-interval digraphs can be seen 
as bigraph and digraph analogs of 
interval graphs, respectively; see also \cite{HHLM20-SIDMA}.
When interval graphs are viewed as reflexive and symmetric digraphs, 
where each vertex has a loop and 
each edge $\{u, v\}$ is replaced by two opposite arcs $(u, v)$ and $(v, u)$, 
they form a subclass of signed-interval digraphs.
Furthermore, the class of signed-interval digraphs 
includes co-TT graphs (i.e., complements of threshold tolerance graphs) and adjusted-interval digraphs 
as subclasses \cite{HHMR20-SIDMA}, 
implying that the following algorithms can be readily adapted to show that 
OSI is solvable in polynomial time 
for co-TT graphs with proper orderings \cite{MRT88-JGT} and adjusted-interval digraphs with min orderings \cite{FHHR12-DAM}.
%
\ifdefined\APP
Due to space constraints, 
the algorithm for signed-interval digraphs and the proof of 
Theorem~\ref{thm:algo signed-interval digraphs} 
are omitted (see Appendix~\ref{sec:apx_min_orderable}).
\fi

\paragraph*{2DOR graphs with comparability weak orderings}
A bipartite graph $G = (V_G, E_G)$ with bipartition $(X_G, Y_G)$ is called 
a \emph{two-directional orthogonal ray graph} 
(or \emph{2DOR graph}) \cite{STU10-DAM} if there exist 
a family of rightward rays (i.e., half-lines) $R_x$ for $x \in X_G$ and 
a family of downward rays $R_y$ for $y \in Y_G$ 
in the plane such that 
for any $x \in X_G$ and $y \in Y_G$, 
$\{x, y\} \in E_G$ if and only if $R_x$ intersects $R_y$. 
2DOR graphs were originally introduced to facilitate 
the defect-tolerant design of nano-circuits \cite{STU09-ISCAS} and 
are equivalent to several other graph classes studied in various contexts, 
including complements of circular-arc graphs with clique cover number two \cite{Spinrad88-JCTSB,FHH99-Comb} 
and interval containment bigraphs \cite{Huang06-JCTSB}; 
see \cite{STU10-DAM,HHMR20-SIDMA} for further details. 

The class of 2DOR graphs admits a characterization 
known as a weak ordering \cite{STU10-DAM}. 
A pair of linear orders, $\sigma_X$ on $X_G$ and $\sigma_Y$ on $Y_G$, 
is called a \emph{weak ordering} if 
for every $x_1, x_2 \in X_G$ and $y_1, y_2 \in Y_G$ 
with $x_1 \prec_{\sigma_X} x_2$ and $y_1 \prec_{\sigma_Y} y_2$, 
$\{x_1, y_2\} \in E_G$ and $\{x_2, y_1\} \in E_G$ imply $\{x_1, y_1\} \in E_G$. 
A bipartite graph is a 2DOR graph if and only if 
it admits a weak ordering. 
In this paper, we also refer to a vertex ordering $\sigma$ of $G$ 
as a weak ordering if 
the pair of suborderings induced by $X_G$ and $Y_G$ forms a weak ordering of $G$. 
Such orderings have also been discussed in~\cite{Huang18-DAM}. 

Now, we introduce the comparability weak ordering. 
A vertex ordering $\sigma$ of a graph $G = (V_G, E_G)$ 
is called a \emph{comparability ordering} if 
for every triple $u, v, w \in V_G$ with $u \prec_{\sigma} v \prec_{\sigma} w$, 
$\{u, v\} \in E_G$ and $\{v, w\} \in E_G$ imply $\{u, w\} \in E_G$; 
see \cite{Damaschke90-incollection,Wood06} and \cite[Section 7.4]{GraphClasses}. 
For a bipartite graph $G$, we refer to its vertex ordering 
as a \emph{comparability weak ordering} 
if it is both a comparability ordering and a weak ordering of $G$. 
Note that every 2DOR graph admits a comparability weak ordering, 
which is given by the concatenation of $\sigma_X$ and $\sigma_Y$. 
Moreover, a bipartite graph is a 2DOR graph 
if and only if it admits a comparability weak ordering. 

The input to \textsc{OSI} 
consists of two bipartite graphs $G$ and $H$ 
with comparability weak orderings $\sigma = (v_1, \dots, v_{n_G})$ and $\tau = (u_1, \dots, u_{n_H})$, respectively. 
Our algorithm determines whether $H$ is ordered subgraph-isomorphic to $G$ and, if so, outputs an ordered subgraph isomorphism $f$. 
It proceeds in the following three steps.
\begin{enumerate}
\item 
Initialize an order-preserving mapping $f$ by setting $f(u_i) \coloneqq v_{i}$ for each $i$ with $1 \le i \le n_H$.
\item 
If $f$ is a subgraph isomorphism from $H$ to $G[\{f(u_i) \mid 1 \le i \le n_H\}]$, 
output $f$ and stop. Otherwise, proceed to the next step. 
\item 
Let $u_p$ and $u_q$ be an arbitrary pair of vertices of $H$
with $p < q$ such that $\{u_p, u_q\} \in E_H$ but $\{f(u_p), f(u_q)\} \notin E_G$; 
such indices exist as $f$ is not a subgraph isomorphism. 
Let $v_{p'} = f(u_p)$ and $v_{q'} = f(u_q)$. 
We consider the following two cases:
\begin{description}
\item[Case~1:]
If $\{v_k, v_{q'}\} \notin E_G$ for any $k$ with $p' < k < q'$, 
update $f$ as follows:
\begin{enumerate}
\item If $q'+1 \le n_G$, set $f(u_q) \coloneqq v_{q'+1}$; otherwise, output No and stop.
\item If $f(u_{q+1}) = v_{q'+1}$, update $q \coloneqq q+1$ and $q' \coloneqq q'+1$, then return to Step~3(a) in Case~1. Otherwise, return to Step~2.
\end{enumerate}
\item[Case~2:]
If there exists some $k$ with $p' < k < q'$ 
such that $\{v_k, v_{q'}\} \in E_G$, 
update $f$ as follows:
\begin{enumerate}
\item If $p'+1 \le n_G$, set $f(u_p) \coloneqq v_{p'+1}$; otherwise, output No and stop.
\item If $f(u_{p+1}) = v_{p'+1}$, update $p \coloneqq p+1$ and $p' \coloneqq p'+1$, then return to Step~3(a) in Case~2. Otherwise, return to Step~2.
\end{enumerate}
\end{description}
\end{enumerate}

\begin{thm}\label{thm:algo 2DOR graphs}
\textsc{OSI} 
on 2DOR graphs with comparability weak orderings 
can be solved in $\mathrm{O}(n_Gn_H^3 + n_G^2n_H)$ time, 
where $n_G = |V_G|$ and $n_H = |V_H|$. 
\end{thm}
\begin{proof}
In the algorithm, let $f_h$ denote the mapping from $V_H$ to $V_G$ after the $h$th iteration of Step~3 (i.e., before the $(h+1)$th iteration of Step~2).
The mapping $f_h$ is injective and order-preserving at each iteration.
To prove the theorem, it suffices to show that 
the algorithm finds an ordered subgraph isomorphism $f$ 
from $H$ to $G$ if one exists. 
Suppose that the algorithm fails to find $f$.
Then, there exists an index $h$ such that 
$f_h(u_i) \preceq_\sigma f(u_i) $ for every $i$ whereas 
$f(u_j) \prec_\sigma f_{h+1}(u_j)$ for some $j$.
Since Step~3 is executed in the $(h+1)$th iteration, 
there must be a pair of vertices $u_p$ and $u_q$ of $H$ 
with $p < q$ such that 
$\{u_p, u_q\} \in E_H$ and $\{f_h(u_p), f_h(u_q)\} \notin E_G$. 
Let $v_{p'} = f_h(u_p)$ and $v_{q'} = f_h(u_q)$. 
\par
Suppose that Case~1 is executed in the $(h+1)$th iteration of Step~3. 
We first prove that $f(u_q) = f_h(u_q)$. 
Since $f_h(u_j) \preceq_\sigma f(u_j) \prec_\sigma f_{h+1}(u_j)$ 
and the algorithm shifts only one vertex per iteration, 
we have $f(u_j) = f_h(u_j)$. 
If $q = j$, the claim holds. 
Otherwise (i.e., $q < j$), 
since $f_{h+1}(u_{j-1}) = f_h(u_j)$ and $f(u_{j-1}) \prec_\sigma f(u_j)$, 
we obtain $f(u_{j-1}) \prec_\sigma f_{h+1}(u_{j-1})$, 
which again implies $f(u_{j-1}) = f_h(u_{j-1})$. 
By applying this argument iteratively, we derive
$f(u_{j-2}) = f_h(u_{j-2})$, 
$f(u_{j-3}) = f_h(u_{j-3})$, $\cdots$, 
and finally, $f(u_q) = f_h(u_q)$. 
However, 
when Case~1 is executed, 
$\{v_k, v_{q'}\} \notin E_G$ for any $k$ with $p' \le k < q'$, 
whereas $\{f(u_p), f(u_q)\} \in E_G$ and 
$f_h(u_p) \preceq_\sigma f(u_p) \prec_\sigma f(u_q) = f_h(u_q)$, 
a contradiction. 
\par
Suppose that Case~2 is executed in the $(h+1)$th iteration of Step~3. 
As in Case~1, we have $f(u_p) = f_h(u_p)$. 
Since Case~2 is executed, 
there exists an index $k$ with $p' < k < q'$ such that 
$\{v_k, v_{q'}\} \in E_G$. 
By the property of comparability orderings, 
$\{v_{p'}, v_{q'}\} \notin E_G$ implies $\{v_{p'}, v_k\} \notin E_G$. 
Thus, graph $G$ contains the pattern shown in Figure~\ref{fig:2dorg:Case2}. 
Without loss of generality, 
assume $f(u_p) \in X_G$ and $f(u_q) \in Y_G$. 
If $v_k \in X_G$ and $v_{q'} \in Y_G$, then 
$\{v_{p'}, v_{q'}\} \notin E_G$ contradicts the property of weak orderings: $\{f(u_p), f(u_q)\} \in E_G$ and $\{v_k, v_{q'}\} \in E_G$ imply $\{v_{p'}, v_{q'}\} \in E_G$. 
Similarly, 
if $v_k \in Y_G$ and $v_{q'} \in X_G$, then 
$\{v_{p'}, v_k\} \notin E_G$ contradicts the property of weak orderings. 
Therefore, the algorithm finds 
the ordered subgraph isomorphism $f$.
\begin{figure}[t]
    \centering\begin{tikzpicture}
\input{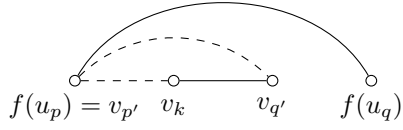}
\draw [dashed] (x) to (y);
\draw [dashed] (x) to [out=45, in=135] (z);
\draw [] (x) to [out=60, in=120] (w);
\draw [] (y) to (z);
\end{tikzpicture}
    \caption{The pattern appearing when Case~2 is executed in the $(h+1)$th iteration of Step~3. (2DOR graphs)}
    \label{fig:2dorg:Case2}
\end{figure}
\par
Steps~1 and 2 run in $\mathrm{O}(n_H)$ and $\mathrm{O}(n_H^2)$ time, respectively.
Step~3 requires $\mathrm{O}(n_G)$ time to decide which case to execute.
Since Steps~2 and 3 can be repeated at most $n_Gn_H$ times, 
the algorithm runs in $\mathrm{O}(n_Gn_H^3 + n_G^2n_H)$ time in total.
\end{proof}

We note that a close analysis of the proof shows that 
the algorithm works correctly even if
the vertex ordering $\tau$ of graph $H$ is arbitrary.

\paragraph*{Signed-interval digraphs with min orderings}
A directed graph (or \emph{digraph}) $G = (V_G, E_G)$, 
which may contain loops and opposite arcs, 
is called a \emph{signed-interval digraph} \cite{HHMR20-SIDMA} if 
it admits a vertex ordering $\sigma = (v_1, \dots, v_{n_G})$ such that 
for any two arcs $(v_i, v_j), (v_{i'}, v_{j'}) \in E_G$, 
there exists an arc $(v_{\min\{i, i'\}}, v_{\min\{j, j'\}}) \in E_G$. 
Such an ordering is called a \emph{min ordering} of $G$. 

\ifdefined\FV
The input to \textsc{OSI} 
consists of two digraphs $G$ and $H$ 
with min orderings $\sigma = (v_1, \dots, v_{n_G})$ and $\tau = (u_1, \dots, u_{n_H})$, respectively. 
The algorithm determines whether $H$ is ordered subgraph-isomorphic to $G$ and, if so, outputs an ordered subgraph isomorphism $f$. 
It proceeds in the following three steps.
\begin{enumerate}
\item 
Initialize an order-preserving mapping $f$ by setting $f(u_i) \coloneqq v_{i}$ for each $i$ with $1 \le i \le n_H$.
\item 
If $f$ is a subgraph isomorphism from $H$ to $G[\{f(u_i) \mid 1 \le i \le n_H\}]$, 
output $f$ and stop. Otherwise, proceed to the next step. 
\item 
Let $u_p$ and $u_q$ be an arbitrary pair of vertices of $H$ 
with $p \le q$ such that at least one of the following holds:
\begin{itemize}
\item $(u_p, u_q) \in E_H$ and $(f(u_p), f(u_q)) \notin E_G$, or 
\item $(u_q, u_p) \in E_H$ and $(f(u_q), f(u_p)) \notin E_G$.
\end{itemize}
Such indices exist as $f$ is not a subgraph isomorphism. 
Note that $p = q$ is possible, as the graph may have loops. 
Let $v_{p'} = f(u_p)$ and $v_{q'} = f(u_q)$. 
We consider the following two cases:
\begin{description}
\item[Case~1:] If either of the following conditions holds:
\begin{itemize}
\item 
$(v_k, v_{q'}) \notin E_G$ for any $k$ with $p' < k < q'$ 
if $(u_p, u_q) \in E_H$ and $(v_{p'}, v_{q'}) \notin E_G$, or 
\item 
$(v_{q'}, v_k) \notin E_G$ for any $k$ with $p' < k < q'$ 
if $(u_q, u_p) \in E_H$ and $(v_{q'}, v_{p'}) \notin E_G$, 
\end{itemize}
update $f$ as follows:
\begin{enumerate}
\item If $q'+1 \le n_G$, set $f(u_q) \coloneqq v_{q'+1}$; otherwise, output No and stop.
\item If $f(u_{q+1}) = v_{q'+1}$, update $q \coloneqq q+1$ and $q' \coloneqq q'+1$, then return to Step~3(a) in Case~1. Otherwise, return to Step~2.
\end{enumerate}
\item[Case~2:] If either of the following conditions holds:
\begin{itemize}
\item 
there exists some $k$ with $p' < k < q'$ 
such that $(v_k, v_{q'}) \in E_G$ 
if $(u_p, u_q) \in E_H$ and $(v_{p'}, v_{q'}) \notin E_G$, or 
\item 
there exists some $k$ with $p' < k < q'$ 
such that $(v_{q'}, v_k) \in E_G$
if $(u_q, u_p) \in E_H$ and $(v_{q'}, v_{p'}) \notin E_G$, 
\end{itemize}
update $f$ as follows:
\begin{enumerate}
\item If $p'+1 \le n_G$, set $f(u_p) \coloneqq v_{p'+1}$; otherwise, output No and stop.
\item If $f(u_{p+1}) = v_{p'+1}$, update $p \coloneqq p+1$ and $p' \coloneqq p'+1$, then return to Step~3(a) in Case~2. Otherwise, return to Step~2.
\end{enumerate}
\end{description}
\end{enumerate}

\fi

\ifdefined\APP
\begin{restatable}[\restateref{thm:algo signed-interval digraphs}]{thm}{thmAlgoSignedIntervalDigraphs}\label{thm:algo signed-interval digraphs}
\textsc{OSI} 
on signed-interval digraphs with min orderings 
can be solved in $\mathrm{O}(n_Gn_H^3 + n_G^2n_H)$ time, 
where $n_G = |V_G|$ and $n_H = |V_H|$. 
\end{restatable}
\else
\begin{thm}\label{thm:algo signed-interval digraphs}
\textsc{OSI} 
on signed-interval digraphs with min orderings 
can be solved in $\mathrm{O}(n_Gn_H^3 + n_G^2n_H)$ time, 
where $n_G = |V_G|$ and $n_H = |V_H|$. 
\end{thm}

\fi
\ifdefined\FV
\ifdefined\APP
\thmAlgoSignedIntervalDigraphs*\label{thm:algo signed-interval digraphs*}
\fi

\begin{proof}
In the algorithm, let $f_h$ denote the mapping from $V_H$ to $V_G$ after the $h$th iteration of Step~3 (i.e., before the $(h+1)$th iteration of Step~2).
The mapping $f_h$ is injective and order-preserving at each iteration.
To prove the theorem, it suffices to show that 
the algorithm finds an ordered subgraph isomorphism $f$ 
from $H$ to $G$ if one exists. 
Suppose that the algorithm fails to find $f$.
Then, there exists an index $h$ such that 
$f_h(u_i) \preceq_\sigma f(u_i) $ for every $i$ whereas 
$f(u_j) \prec_\sigma f_{h+1}(u_j)$ for some $j$.
Since Step~3 is executed in the $(h+1)$th iteration, 
there must be a pair of vertices $u_p$ and $u_q$ of $H$ with $p \le q$ such that 
$(u_p, u_q) \in E_H$ and $(f_h(u_p), f_h(u_q)) \notin E_G$ 
or 
$(u_q, u_p) \in E_H$ and $(f_h(u_q), f_h(u_p)) \notin E_G$. 
Let $v_{p'} = f_h(u_p)$ and $v_{q'} = f_h(u_q)$. 
\par
Suppose that Case~1 is executed in the $(h+1)$th iteration of Step~3. 
As with 2DOR graphs, we have $f(u_q) = f_h(u_q)$. 
If $p = q$ then $u_q$ has a loop while $f_h(u_q)$ does not, 
contradicting the assumption that $f$ is a subgraph isomorphism.
Now, suppose $p < q$. 
Then, $f_h(u_p) \preceq_\sigma f(u_p) \prec_\sigma f(u_q) = f_h(u_q)$. 
Since Case~1 is executed, 
$(v_k, v_{q'}) \notin E_G$ for any $k$ with $p' \le k < q'$ whereas $(f(u_p), f(u_q)) \in E_G$, or 
$(v_{q'}, v_k) \notin E_G$ for any $k$ with $p' \le k < q'$ whereas $(f(u_q), f(u_p)) \in E_G$, 
a contradiction. 
\par
Suppose that Case~2 is executed in the $(h+1)$th iteration of Step~3. 
As with 2DOR graphs, we have $f(u_p) = f_h(u_p)$. 
If $p = q$, we reach a contradiction as in Case 1. 
Now, suppose $p < q$. 
Since Case~2 is executed, 
there exists an index $k$ with $p' < k < q'$ such that 
$(v_k, v_{q'}) \in E_G$ or $(v_{q'}, v_k) \in E_G$. 
Therefore, graph $G$ contains one of the patterns shown in Figure~\ref{fig:signed:Case2}. 
In the former case, 
$(v_{p'}, v_{q'}) \notin E_G$ contradicts the property of min orderings: $(f(u_p), f(u_q)) \in E_G$ and $(v_k, v_{q'}) \in E_G$ imply $(v_{p'}, v_{q'}) \in E_G$. 
Similarly, 
in the latter case, 
$(v_{q'}, v_{p'}) \notin E_G$ contradicts the property of min orderings. 
Therefore, the algorithm finds 
the ordered subgraph isomorphism $f$.
\begin{figure}[t]
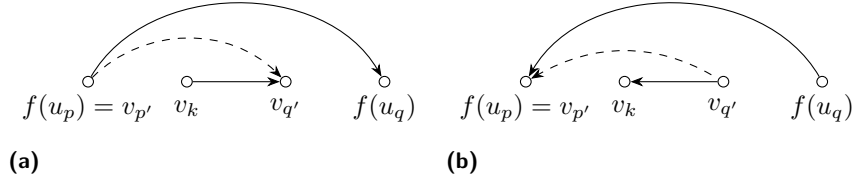

  \centering
  \subcaptionbox{}{\begin{tikzpicture}
\input{./figs/template_p4}
\draw [dashed,-Stealth] (x) to [out=45, in=135] (z);
\draw [-Stealth] (x) to [out=60, in=120] (w);
\draw [-Stealth] (y) to (z);
\end{tikzpicture}}
  \subcaptionbox{}{\begin{tikzpicture}
\input{./figs/template_p4}
\draw [dashed,shorten <=3pt,-Stealth] (z) to [out=150, in=30] (x);
\draw [-Stealth] (w) to [out=120, in=60] (x);
\draw [-Stealth] (z) to (y);
\end{tikzpicture}}
  \caption{
    The pattern appearing when Case~2 is executed in the $(h+1)$th iteration of Step~3. (Signed-interval digraphs)
    (a) The case when $(u_p, u_q) \in E_H$ and $(f_h(u_p), f_h(u_q)) \notin E_G$.
    (b) The case when $(u_q, u_p) \in E_H$ and $(f_h(u_q), f_h(u_p)) \notin E_G$.
  } 
  \label{fig:signed:Case2}
\end{figure}
\par
As with 2DOR graphs, 
the algorithm runs in $\mathrm{O}(n_Gn_H^3 + n_G^2n_H)$ time.
\end{proof}

As with 2DOR graphs, 
the vertex ordering $\tau$ of graph $H$ 
does not necessarily need to be a min ordering. 

\fi

\subsection{MCOIS on Threshold/Chain/Cochain Graphs}
\label{sec:alg-threshold}

In Section~\ref{sec:hard-threshold}, we proved the NP-hardness for OISI on threshold graphs, chain graphs, and cochain graphs. 
In this section, we propose polynomial-time algorithms for a more general problem, MCOIS, on these graph classes, assuming that the input vertex orderings are restricted to inclusion orderings.
We would like to note that the restriction on vertex orderings for these algorithms is natural in the sense that inclusion orderings characterize these graph classes. 
Let $G=(X\cup Y, E)$ with a bipartition $(X, Y)$, where $X$ (and $Y$) admits an inclusion ordering. 
Recall that an \emph{inclusion ordering} of $G$ is defined as the inclusion ordering on $Y$ followed by the inclusion ordering on $X$. 
We first present a polynomial-time algorithm for MCOIS on threshold graphs with inclusion orderings, and then show that the same approach applies to chain graphs and cochain graphs.

Let $G=(X_G\cup Y_G, E_G)$ and $H=(X_H\cup Y_H, E_H)$ be threshold graphs and $(y^G_1,$ \linebreak$\dots, y^G_{|Y_G|}, x^G_1, \dots, x^G_{|X_G|})$ and $(y^H_1, \dots, y^H_{|Y_H|}, x^H_1, \dots, x^H_{|X_H|})$ be inclusion orderings of $G$ and $H$, respectively. 
We assume that $X_G$ and $X_H$ are both cliques (or independent sets) in $G$ and $H$, respectively; otherwise, the answer is trivial. 
The vertices in $X_G$ are matched with at most one vertex in $Y_H$, and those in $Y_G$ are matched with at most one vertex in $X_H$. 
By guessing the match, we only consider the case that the vertices in $X_G$ (resp. $Y_G$) are matched with the vertices in $X_H$ (resp. $Y_H$) in an optimal solution $(Z, f_{G}, f_{H})$, where $f_{G}$ and $f_{H}$ are ordered induced subgraph isomorphisms from $Z$ to $G$ and $H$, respectively. 
Such an optimal solution satisfies one of the followings: (1) no vertex in $X_G$ is matched with a vertex in $X_H$, or (2) there exists a pair of the largest indices $(i, j)$ such that $f^{-1}_G(x^G_i) = f^{-1}_H(x^H_j)$. 
In case (1), we match as many vertices of $Y_G$ and $Y_H$ as possible; that is, we obtain $\min\{|Y_G|, |Y_H|\}$. 
In case (2), the vertices $Y_G \setminus N_{Y_G}(x^G_{i})$ are matched with the vertices $Y_H\setminus N_{Y_H}(x^H_{j})$ because they are isolated in $G[Y_G \cup \{x^G_i\}]$ and $H[Y_H \cup \{x^H_j\}]$. 
Thus, the number of vertices of the optimal solution is the sum of (i) the number of vertices of the connected ordered induced subgraph of $G_i\coloneqq G[X^G_i\cup N_{Y_G}(X^G_i)]$ and $H_j\coloneqq H[X^H_j\cup N_{Y_H}(X^H_j)]$, and (ii) the smaller of $|Y_G \setminus N_{Y_G}(x^G_{i})|$ and $|Y_H\setminus N_{Y_H}(x^H_{j})|$, where $X^G_i = \{x^G_1, \dots, x^G_{i}\}$ and $X^H_j = \{x^H_1, \dots, x^H_{j}\}$. 

It remains for us to compute (i) the size of a connected optimal ordered induced subgraph of $G_i$ and $H_j$ using dynamic programming. 
For $i\in [|X_G|]\cup\{0\}, j\in [|X_H|]\cup \{0\}$, the entry of $\dptable[i, j]$ represents the number of vertices of an optimal solution for $G_i$ and $H_j$ when $x^G_i$ is matched with $x^H_j$. 
If $i = 0$ or $j = 0$, then $\dptable[i, j] = 0$ since there is no vertex in $G_i$ or $H_j$. 
For $i\in [|X_G|]$ and $j\in [|X_H|]$, let a triple $T = (Z_{ij}, f_{G_i}, f_{H_j})$ be an optimal solution for $G_i$ and $H_j$ such that $x^G_{i}$ is matched with $x^H_{j}$. 
Let $i'<i$ and $j'<j$ be the largest indices such that $f_{G_i}^{-1}(x^G_{i'}) = f_{H_j}^{-1}(x^H_{j'})$. 
The triple $T$ can be computed from $(Z_{i' j'}, f_{G_{i'}}, f_{H_{j'}})$ of an optimal solution for $G_{i'}$ and $H_{j'}$ as follows; 
$x^G_{i}$ is matched with $x^H_{j}$, and the vertices in $N_{Y_G}(x^G_{i})\setminus N_{Y_G}(x^G_{i'})$ are matched with the vertices in $N_{Y_H}(x^H_{j})\setminus N_{Y_H}(x^H_{j'})$, as they have not been matched yet and can be matched. 
To consider the case that there is no such pair $(i', j')$, we assume that $N_{Y_G}(x^G_{0}) = N_{Y_H}(x^H_{0}) = \emptyset$ for $i'=j'=0$. 
Thus, we compute $\dptable[i, j]$ as follows: 
\[
 \dptable[i, j] = \max_{i'\le i, j'\le j}(\dptable[i', j'] + d_{ij} + 1), 
\]
where $d_{ij} = \min\{|N_{Y_G}(x^G_{i})\setminus N_{Y_G}(x^G_{i'})|, |N_{Y_H}(x^H_{j})\setminus N_{Y_H}(x^H_{j'})|\}$. 
Finally, we obtain the size of the maximum common ordered induced subgraph by computing 
\[
 \max_{i\in [|X_G|], j\in [|X_H|]} (\dptable[i, j] + \min\{|Y_G \setminus N_{Y_G}(x^G_{i})|, |Y_H\setminus N_{Y_H}(x^H_{j})|\}). 
\]

We analyze the time complexity of the algorithm. 
The algorithm first guesses the $n_G\cdot n_H + 1$ cases: $x_G\in X_G$ matches with $y_H\in Y_H$, $y_G\in Y_G$ matches with $x_H\in X_H$, and no such match exists. 
There are $\mathrm{O}(n_G n_H)$ table entries, each of which can be computed in $\mathrm{O}(n_G n_H)$ time.
Thus, the running time to fill the DP table is bounded by $\mathrm{O}(n_G^2 n_H^2)$. 
After computing the DP table, it takes $\mathrm{O}(n_G n_H)$ time to obtain the size of an optimal solution. 
Therefore, the overall time complexity of our algorithm is $\mathrm{O}(n_G^3 n_H^3)$. 

Our algorithm does not use the neighborhood relation within $X$ or $Y$ except for trivial solutions, and the vertex partition can be determined uniquely from the inclusion ordering. 
Thus, the algorithm for threshold graphs can be extended to chain graphs and cochain graphs. 
\begin{thm}\label{thm:tcc:poly-alg}
 MCOIS can be solved in polynomial time on threshold graphs, chain graphs, and cochain graphs with inclusion orderings. 
\end{thm}



\subsection{MCO(I)S Parameterized by Pathwidth}
\label{sec:FPTalg-pathwidth}

\newcommand{\f}{\text{f}}
\renewcommand{\c}{\text{c}}
\renewcommand{\i}{\text{i}}



Let $G = (V_G, E_G)$ be a graph with $n_G$ vertices.
A \emph{path decomposition} of $G$ is a sequence of vertex subsets, called \emph{bags}, $\mathcal{P} = (X_1, X_2, \dots, X_r)$ satisfying: 
(1) $\bigcup_{i\in [r]} X_i = V_G$, 
(2) $\forall \{u, v\}\in E_G, \exists i\in [r], \{u, v\}\subseteq X_i$, and 
(3) for $v \in V_G$, the bags containing $v$ are consecutive.
The \emph{width} of a path decomposition is the size of the largest bag minus~1. 
The \emph{pathwidth} of a graph $G$ is the minimum width over all possible path decompositions of $G$.

A path decomposition $(X_0, X_1, X_2, \dots, X_{2n_G})$ of a graph $G = (V_G, E_G)$ is \emph{nice} if it satisfies the following additional constraints:
(4) $X_0 = X_{2n_G} = \emptyset$, and (5) for every $i\in [2n_G]$, there is either a vertex $x^\i \notin X_{i-1}$ such that $X_i = X_{i-1}\cup \{x^\i\}$, or a vertex $x^\f \in X_i$ such that $X_i = X_{i-1}\setminus\{x^\f\}$.
For a nice path decomposition, if $X_i = X_{i-1}\cup \{x^\i \}$, $X_i$ is called an \emph{introduce node} and $X_i$ \emph{introduces} $x^\i$, and if $X_i = X_{i-1}\setminus\{x^\f \}$, $X_i$ is called a \emph{forget node} and $X_i$ \emph{forgets} $x^\f$. 
Given a path decomposition of $G$ with width $w$, we can compute a nice path decomposition of $G$ with width $w$ in polynomial time~\cite{CyganFKLMPPS15}. 
For a nice path decomposition $\mathcal{P}$ of $G$, a vertex ordering of $\sigma = (x_1, \dots, x_{n_G})$ is called an \emph{introduce ordering} if for any two vertices $x_i, x_j\in V_G$, we have $x_i\prec_{\sigma} x_j$ if and only if the index of the bag where $x_i$ is introduced is smaller than that of $x_j$. 
This ordering is equivalent to the ordering defined in terms of the vertex separation number~\cite{Kinnersley92}. 
Based on this equivalence, we also say that $w$ is the pathwidth of the vertex ordering~$\sigma$.

In this section, we present FPT algorithms for MCO(I)S parameterized by the pathwidth $w$ of the vertex orderings of $G$ and $H$, assuming that their nice path decompositions $\mathcal{P}_G = (X_0, X_1, \dots, X_{2n_G})$ and $\mathcal{P}_H = (Y_0, Y_1, \dots, Y_{2n_H})$ of width at most $w$ and their introduce orderings $\sigma = (x_1, \dots, x_{n_G})$ and $\tau = (y_1, \dots, y_{n_H})$ associated with $\mathcal{P}_G$ and $\mathcal{P}_H$, respectively, are given as input.
The algorithms are based on dynamic programming on the nice path decompositions. 
Let $G_i$ (resp.~$H_i$) be an induced subgraph $G[\bigcup_{j\in [i]} X_j]$ (resp.~$H[\bigcup_{j\in [i]} Y_j]$) for an integer $i\in [2n_G]$ (resp.~$i\in [2n_H]$). 
Note that $G_{2n_G} = G$ and $H_{2n_H} = H$. 

Our algorithms compute an optimal solution $(Z_{ij}, f_{G_i}, f_{H_j})$ of $G_i$ and $H_j$, where $f_{G_i}$ and $f_{H_j}$ are ordered (induced) subgraph isomorphisms from $Z_{ij}$ to $G_i$ and $H_j$, respectively. 
The key idea of the algorithm is to match the introduce vertices $x^\i$ and $y^\i$ only when both $X_i$ and $Y_j$ are introduce nodes.  
It is easy to see that the ordering constraint at the introduced vertices is always satisfied since $x^\i$ and $y^\i$ are the largest in the suborderings induced by $V_{G_i}$ and $V_{H_j}$, respectively. 
We need to determine whether or not we can match $x^\i$ and $y^\i$ for MCOIS or to count the number of additional edges for MCOS when $x^\i$ matches $y^\i$. 
To this end, we maintain four subsets $X^\f, X^\c$ of $X_i$, and $Y^\f, Y^\c$ of $Y_j$ as a state in our dynamic programming. 
The first two subsets indicate that any vertex in $X^\f$ is matched with a \emph{forgotten} vertex, and any vertex in $X^\c$ is matched with a \emph{current} vertex in $Y_j$. 
Analogously, any vertex in $Y^\f$ and any vertex in $Y^\c$ are respectively matched with a forgotten vertex and a current vertex. 
Note that $|X^\c| = |Y^\c|$ has to hold. 
Since $f_{G_i}$ and $f_{H_j}$ are ordered (induced) subgraph isomorphisms, the matching between $X^\c$ and $Y^\c$ is uniquely determined.

More specifically, we maintain the following table. 
For integers $i\in [2n_G]\cup\{0\}$ and $j\in [2n_H]\cup \{0\}$, the value $\dptable[(i, X^\f, X^\c), (j, Y^\f, Y^\c)]$ is the number of edges (vertices) in an optimal solution $(Z_{ij}, f_{G_i}, f_{H_j})$ of $G_i$ and $H_j$ satisfying $X^\f, X^\c, Y^\f$, and $Y^\c$ conditions. 
If there is no $(f_{G_i}, f_{H_j})$, then the entry is $-\infty$. 
Since $X_{2n_G} = Y_{2n_H} = \emptyset$, $G_{2n_G} = G$ and $H_{2n_H} = H$, the entry $\dptable[(2n_G, \emptyset, \emptyset), (2n_H, \emptyset, \emptyset)]$ is stored the number of edges (vertices) in a maximum common ordered (induced) subgraph of $G$ and $H$. 
\ifdefined\APP
We describe the details of the computation of the recursion in Appendix~\ref{sec:pathwidth-recursion}. 
\fi

\ifdefined\FV
We describe how to compute each entry $\dptable[(i, X^\f, X^\c), (j, Y^\f, Y^\c)]$ of the table in dynamic programming.

\noindent {\bf Case $i = 0$ or $j = 0$: }
$\dptable[(i, \emptyset, \emptyset), (j, \emptyset, \emptyset)] = 0$ because there is no vertex in $G_i$ or $H_j$. 
If at least one of $X^\f, X^\c, Y^\f$ or $Y^\c$ is not an empty set, the entry takes $-\infty$. 

\noindent {\bf Case $X_i$ and $Y_j$ are both introduce nodes: }
Let $x^\i$ and $y^\i$ be vertices introduced at $X_i$ and $Y_j$, respectively. 
The recursion of $\dptable[(i, X^\f, X^\c), (j, Y^\f, Y^\c)]$ is as follows: 
\[
  \begin{cases}
   \dptable[(i-1, X^\f, X^\c\setminus\{x^\i\}), (j-1, Y^\f, Y^\c\setminus \{y^\i\})] + z & \text{if } x^\i \in X^\c \text{ and } y^\i\in Y^\c\\ 
   \dptable[(i, X^\f, X^\c), (j-1, Y^\f, Y^\c)] & \text{if } x^\i \in X^\c \text{ and } y^\i \notin Y^\c\\
   \dptable[(i-1, X^\f, X^\c), (j, Y^\f, Y^\c)] & \text{if } x^\i \notin X^\c \text{ and } y^\i \in Y^\c\\
   \dptable[(i-1, X^\f, X^\c), (j-1, Y^\f, Y^\c)] & \text{if } x^\i \notin X^\c \text{ and } y^\i \notin Y^\c\\
  \end{cases}
\]
In the equation of the first line, we match between $x^\i$ with $y^\i$. 
The value of $z$ added to the entry differs depending on the problems, induced or non-induced.

For the ``induced'' problem, we first verify that $x^\i$ and $y^\i$ are not adjacent to any vertex in $X^\f$ and $Y^\f$, respectively. 
If either of them is adjacent, it is impossible to obtain a common induced subgraph because the matched vertex for a vertex in $X^\f$ or $Y^\f$ is forgotten. 
For example, if $x^\i$ is adjacent to $x \in X^\f$, the vertex $y$ matched with $x$ is not adjacent to $y^\i$ because $y$ is forgotten. 
That is, $x^\i$ and $y^\i$ cannot be matched. 
Thus, if $x^\i$ is adjacent to a vertex in $X^\f$ (or similarly for $y^\i$ with $Y^\f$), the entry takes the value $-\infty$.  
Assume that $x^\i$ is not adjacent to any vertex in $X^\f$ and $y^\i$ is not adjacent to any vertex in $Y^\f$. 
We then verify whether the adjacency of $x^\i$ with $X^\c$ and $y^\i$ with $Y^\c$ is valid. 
Since the matching between $X^\c$ and $Y^\c$ is uniquely determined by vertex ordering, we can easily check the validity.  
If $x^\i$ and $y^\i$ can be matched, we set $z = 1$; otherwise, we set $z = -\infty$.

For the ``non-induced'' problem, we count the number of edges by matching $x^\i$ with $y^\i$. 
From the same reason for the ``induced'' problem, we do not count the number of edges between $x^\i$ with $x\in X^\f$ and $y^\i$ with $y\in Y^\f$. 
Thus, for each pair of vertices $x\in X^\c$ and $y\in Y^\c$ satisfying $f_{G_{i-1}}^{-1}(x) = f_{H_{j-1}}^{-1}(y)$, we check whether there are both edges $\{x^\i, x\}$ and $\{y^\i, y\}$ and $z$ is the number of the pairs of edges.

\noindent {\bf Case $X_i$ is an introduce node and $Y_j$ is a forget node: }
Let $x^\i$ and $y^\f$ be vertices introduced and forgotten at $X_i$ and $Y_j$, respectively. 
If $x^\i\in X^\c$ or $x^\i\in X^\f$, the recursion of $\dptable[(i, X^\f, X^\c), (j, Y^\f, Y^\c)]$ is the maximum among the three values: 
\begin{align*}
 &\dptable[(i, X^\f, X^\c), (j-1, Y^\f, Y^\c)], \\
 &\dptable[(i, X^\f\setminus\{x\}, X^\c\cup \{x\}), (j-1, Y^\f, Y^\c\cup \{y^\f\})]\text{, and}\\
 &\dptable[(i, X^\f, X^\c), (j-1, Y^\f\cup \{y^\f\}, Y^\c)].
\end{align*}
Note that integer $i$ in the first triplet does not decrease in the recursion since $x^\i$ is matched to a vertex in $V_{H_{j-1}}$ and is matched at an introduce node $Y_{j'}$ for some $j'<j$. 
The three equations are classified based on whether the forgotten vertex $y^\f$ is matched with a vertex or not. 
The first value 
corresponds to the case where $y^\f$ has no matched vertex in $V_{G_i}$. 
The second value applies when $y^\f$ is matched with a current vertex $x$. 
The third value represents the case where $y^\f$ is matched with a forgotten vertex. 

Next, we consider the case where $x^\i\notin X^\c$ and $x^\i\notin X^\f$. 
The equation $\dptable[(i, X^\f, X^\c), (j, Y^\f, Y^\c)]$ is the maximum among the three values: 
\begin{align*}
   &\dptable[(i-1, X^\f, X^\c), (j-1, Y^\f, Y^\c)], \\
   &\dptable[(i-1, X^\f\setminus\{x\}, X^\c\cup \{x\}), (j-1, Y^\f, Y^\c\cup \{y^\f\})]\text{, and}\\
   &\dptable[(i-1, X^\f, X^\c), (j-1, Y^\f\cup \{y^\f\}, Y^\c)].
\end{align*}
The difference from the equation in the case of $x^\i\in X^\c$ or $x^\i\in X^\f$ is that the integer $i$ is decremented since the introduced vertex $x^\i$ is not matched. 

Similar to the above equations, we can obtain the recurrence formula for the case of a forget node $X_i$ and an introduce node $Y_j$.

\noindent {\bf Case $X_i$ and $Y_j$ are both forget nodes: }
Let $x^\f$ and $y^\f$ be the vertices forgotten at $X_i$ and $Y_j$, respectively. 
We consider the following cases: $x^\f$ (resp. $y^\f$) is not matched, $x^\f$ (resp. $y^\f$) is matched with a forgotten vertex in $\bigcup_{j'=1}^{j-1}Y_{j'}\setminus Y_j$ (resp. $\bigcup_{i'=1}^{i-1}X_{i'}\setminus Y_i$), $x^\f$ (resp. $y^\f$) is matched with a current vertex in $Y_{j-1}$ (resp. $X_{i-1}$), and $x^\f$ is matched with $y^\f$. 
 Thus, the value $\dptable[(i, X^\f, X^\c), (j, Y^\f, Y^\c)]$ is assigned the maximum of the following 10 values: 
\begin{align*}
& \dptable[(i-1, X^\f, X^\c), (j-1, Y^\f, Y^\c)], \\
& \dptable[(i-1, X^\f, X^\c\cup\{x^\f\}), (j-1, Y^\f\setminus \{y\}, Y^\c\cup \{y\})], \\
& \dptable[(i-1, X^\f, X^\c\cup\{x^\f\}), (j-1, (Y^\f\cup \{y^\f\})\setminus \{y\}, Y^\c\cup \{y\})], \\
& \dptable[(i-1, X^\f\setminus \{x\}, X^\c\cup\{x\}), (j-1, Y^\f, Y^\c\cup \{y^\f\})], \\
& \dptable[(i-1, X^\f\setminus \{x\}, X^\c\cup\{x\}), (j-1, Y^\f, Y^\c\cup \{y^\f\})], \\
& \dptable[(i-1, X^\f\setminus \{x\}, X^\c\cup\{x, x^\f\}), (j-1, Y^\f\setminus\{y\}, Y^\c\cup \{y, y^\f\})], \\
& \dptable[(i-1, X^\f, X^\c\cup \{x^\f\}), (j-1, Y^\f, Y^\c\cup \{y^\f\})], \\
& \dptable[(i-1, X^\f, X^\c), (j-1, Y^\f\cup \{y^\f\}, Y^\c)], \\
& \dptable[(i-1, X^\f\cup \{x^\f\}, X^\c), (j-1, Y^\f, Y^\c)], \text{ and} \\
& \dptable[(i-1, X^\f\cup \{x^\f\}, X^\c), (j-1, Y^\f\cup \{y^\f\}, Y^\c)].  
\end{align*}

\fi

To analyze the time complexity of the algorithms, we estimate the size of the table $\dptable[(i, X^\f, X^\c)(j, Y^\f, Y^\c)]$. 
For each $i\in [2n_G]$, there are at most $3^w$ possible ways to form the first triple $(i, \cdot, \cdot)$ because $X_i$ is partitioned into $X^\f, X^\c$, and the remaining vertices. 
Similarly, for each $j\in [2n_H]$, there are at most $3^w$ possible ways to form the second triple $(j, \cdot, \cdot)$. 
Thus, the number of possible entries in the DP table is $\mathrm{O}(9^w)$. 
Each entry can be computed in time polynomial in $w$. 

Here, we carefully observe the algorithm for the non-induced problem, MCOS. 
A crucial difference from the MCOIS algorithm is that we do not take into account the edges between $x^\i$ and $X^\f$, nor $y^\i$ and $y^\f$, respectively. 
Thus, we do not need to maintain $X^\f$ and $Y^\f$. 
This improves the size of the DP table to $\mathrm{O}(4^w)$. 

\begin{thm}\label{thm:FPT_pathdecomposition}
Let $\mathcal{P}_G$ and $\mathcal{P}_H$ be path decompositions of graphs $G$ and $H$ with width at most $w$, and $\sigma$ and $\tau$ be introduce orderings of $G$ and $H$, respectively. 
We can solve MCOIS and MCOS in time $\mathrm{O}^*(9^w)$ and $\mathrm{O}^*(4^w)$, respectively.\footnote{The $\mathrm{O}^*$ notation supresses polynomial factors in the input size.}
\end{thm}

\subsection{MCO(I)S Parameterized by Vertex Cover Number}

\label{sec:FPTalg-VC}
We now show that both MCOS and MCOIS are fixed-parameter tractable parameterized by vertex cover number of both input graphs,
even if there is no restriction on vertex orderings.
\begin{thm}
\label{thm:vc}
\textsc{MCOS} and \textsc{MCOIS} admit $\mathrm{O}^{*}(2^{p^{2}})$-time algorithms, where $p$ is the maximum of vertex cover numbers of input graphs.
\end{thm}

The algorithms for both problems are based on almost the same dynamic programming approach.
\ifdefined\FV
Thus, we first present the algorithm for MCOS and then explain how we can modify the algorithm for MCOIS\@.
\fi
\ifdefined\APP
Thus, we present the algorithm for MCOS only. A straightforward modification for MCOIS is omitted (see Section~\ref{sec:apx_vc}).
\fi

Let $G$ and $H$ be the input graphs of MCOS with vertex orderings
$(u_{1}, u_{2}, \dots,\linebreak[1] u_{|V(G)|})$ and $(v_{1}, v_{2}, \dots, v_{|V(H)|})$, respectively.
Let $U_{\le i} = \{u_{1}, \dots, u_{i}\}$ and $U_{> i} = \{u_{i+1}, \dots, u_{|V(G)|}\}$.
We define $V_{\le j}$ and $V_{> j}$ analogously.
Our goal is to find a subgraph $Z$ of $G$ with maximum number of edges that is ordered subgraph-isomorphic to $H$.

We assume that a vertex cover $S$ of $G$ with size at most $p$ is given and $Z$ contains all vertices in $S$.
We can justify this assumption by first finding a minimum vertex cover $S'$ of $G$ by an FPT algorithm parameterized by $p$~\cite{CyganFKLMPPS15},
trying all $2^{|S'|}$ subsets $S$ of $S'$, and then updating $G$ as $G \coloneqq G - (S'\setminus S)$.
For simplicity, we assume that $|S| = p$ in the following.

Let $T_{1}, \dots, T_{q}$ be the twin classes of $H$, that is, 
each $T_i$ is a maximal vertex such that $N_{H}(u) \setminus \{v\} = N_{H}(v) \setminus \{u\}$ holds for any two vertices $u, v \in T_{i}$.
These classes can be computed in polynomial time.
Note that $q \le 2^{p} + p$ as $H$ has vertex cover number at most~$p$.

\smallskip\textbf{Definition of the DP table.}
The key idea of the algorithm is to guess for each $s \in S$ the type of its image in $H$.
That is, we guess a mapping $\phi \colon S \to [q]$ such that $f(s) \in T_{\phi(s)}$ for each $s \in S$,
where $f$ is an ordered subgraph isomorphism from a subgraph of $G$ to $H$.

Now we define the DP table $\dptable_{\phi}$.
For $i \in [n_G] \cup \{0\}$ and $j \in [n_H] \cup \{0\}$,
the entry $\dptable_{\phi}[i,j]$ stores the maximum number of edges in an ordered subgraph $Z$ of $G$ with $S \subseteq V(Z) \subseteq S \cup U_{\le i}$ such that
there is an ordered subgraph isomorphism $f$ from $Z$ to $H$ that satisfies
\begin{itemize}
  \item $f((S \cap U_{\le i}) \cup (V(Z) \setminus S)) \subseteq V_{\le j}$,
  \item $f(S \cap U_{> i}) \subseteq V_{> j}$, and
  \item $f(s) \in T_{\phi(s)}$ for each $s \in S$,
\end{itemize}
where we set $\dptable_{\phi}[i,j] = -\infty$ if no subgraph of $G$ satisfies these conditions.
When $\dptable_{\phi}[i,j] \ne -\infty$, we call a pair $(Z,f)$ satisfying the conditions above
and having the maximum number of edges $|E(Z)| = \dptable_{\phi}[i,j]$ a \emph{certificate} for $\dptable_{\phi}[i,j]$.

Observe that, when $(i,j) = (|V(G)|, |V(H)|)$, the conditions above boil down to the last one, 
which asks that the ordered subgraph isomorphism $f$ respects $\phi$.
Thus, for finding the maximum size of an ordered common subgraph of $G$ and $H$ under $\phi$,
it suffices to compute $\dptable_{\phi}[|V(G)|, |V(H)|]$.

\smallskip\textbf{Computing the DP table entries.}
\textit{Preprocessing.}
We first test for each pair $(i,j)$ whether $\dptable_{\phi}[i,j] \ne -\infty$.
Let $s_{1}, \dots, s_{|S|}$ be the vertices in $S$ ordered in this way in $G$.
We check the existence of $|S|$ vertices $w_{1}, \dots, w_{|S|} \in V(H)$ ordered in this way in $H$ such that
\begin{itemize}
  \item $w_{h} \in T_{\phi(s_{h})}$ for each $s_h \in S$,
  \item $\{w_{h} \mid 1 \le h \le |S \cap U_{\le i}|\} \subseteq V_{\le j}$, and
  \item $\{w_{h} \mid |S \cap U_{\le i}| + 1 \le h \le |S|\} \subseteq V_{> j}$.
\end{itemize}
If such vertices exist, we can greedily find them by taking the first vertex satisfying the conditions for each $h$.
Such vertices guarantee that $\dptable_{\phi}[i,j] \ge 0$ since we can set $Z = (S, \emptyset)$ and set $f(s_{h}) = w_{h}$ for each $h$.
If such vertices do not exist, then we can conclude that $\dptable_{\phi}[i,j] = -\infty$.

\textit{Case 1:}
$\min(i,j) = 0$ and $\dptable_{\phi}[i,j] \ne -\infty$.
In this case, the definition of $\dptable_{\phi}$ implies that $V(Z) = S$ as follows.
If $i=0$, then $S \subseteq V(Z) \subseteq S \cup U_{\le 0}$ implies $S = V(Z)$ as $U_{\le 0} = \emptyset$.
If $j=0$, then $f((S \cap U_{\le i}) \cup (V(Z) \setminus S)) \subseteq V_{\le j}$ implies $V(Z) \setminus S = \emptyset$ as $V_{\le j} = \emptyset$,
and thus $V(Z) = S$ as $S \subseteq V(Z)$.
Since $\dptable_{\phi}[i,j] \ne -\infty$, there exists an ordered subgraph isomorphism from a subgraph of $G[S]$ to $H$ respecting $\phi$.
Thus, $\dptable_{\phi}[i,j]$ can be computed by counting the number of edges $\{s, s'\}$ in $G[S]$
such that the type $\phi(s)$ vertices and the type $\phi(s')$ vertices are adjacent.

\textit{Case 2:} 
$\min(i,j) \ne 0$, $u_{i} \in S$, and $\dptable_{\phi}[i,j] \ne -\infty$.
Let $(Z, f)$ be a certificate for $\dptable_{\phi}[i,j]$.
Recall that $Z$ is an ordered subgraph of $G$ and $f$ is an ordered subgraph isomorphism from $Z$ to $H$.
Since $u_{i} \in S$, its image $v_{h} \coloneqq f(u_{i})$ belongs to $V_{\le j}$.
Furthermore, $v_{h} \in T_{\phi(u_{i})}$ holds since $f$ obeys $\phi$.
We can see that the pair $(Z, f)$ satisfies all conditions for $\dptable_{\phi}[i-1, h-1]$ as well.
On the other hand, for a pair $(Z', f')$ that satisfies the conditions for $\dptable_{\phi}[i-1, h-1]$,
it is always possible to change $f'(u_{i})$ to $v_{h}$ without changing $Z'$
since $f'(u_{i}) \in V_{> h-1}$.
Hence, 
\[
  \dptable_{\phi}[i,j]
  =
  \max\{\dptable_{\phi}[i-1, h-1] \mid h \le j, \, v_{h} \in T_{\phi(u_{i})}\}.
\]

\textit{Case 3:} 
$\min(i,j) \ne 0$, $u_{i} \notin S$, and $\dptable_{\phi}[i,j] \ne -\infty$.
Let $(Z, f)$ be a certificate for $\dptable_{\phi}[i,j]$.
Since $u_{i} \notin S$, we need to consider two subcases of $u_{i} \notin V(Z)$ and $u_{i} \in V(Z)$, and then take the maximum of them.
If $u_{i} \notin V(Z)$, we can simply set $\dptable_{\phi}[i,j] = \dptable_{\phi}[i-1, j]$.
Now assume that $u_{i} \in V(Z)$ and $f(u_{i}) = v_{h}$ for some $h \le j$.
Let $f'$ be the restriction of $f$ to $V(Z) \setminus \{u_{i}\}$.
Observe that the pair $(Z - u_{i}, f')$ satisfies all conditions for $\dptable_{\phi}[i-1, h-1]$, and thus it achieves the maximum edge number $\dptable_{\phi}[i-1, h-1]$.
Since this pair $(Z - u_{i}, f')$ for $\dptable_{\phi}[i-1, h-1]$ is extended to the pair $(Z, f)$ for $\dptable_{\phi}[i, j]$ by setting $f(u_{i}) = v_{h}$,
the increase $\dptable_{\phi}[i, j] - \dptable_{\phi}[i-1, h-1]$ can be computed by counting the number of edges $\{u_{i}, s\}$ with $s \in S$
such that $v_{h}$ is adjacent to the type $\phi(s)$ vertices.
Thus, by setting $\mu(i,h)$ to this number, we can write $\dptable_{\phi}[i,j] = \dptable_{\phi}[i-1, h-1] + \mu(i, h)$.
The discussion so far implies that
\[
  \dptable_{\phi}[i,j]
  =
  \max\{
    \dptable_{\phi}[i-1, j], 
    \max\{\dptable_{\phi}[i-1, h-1] + \mu(i, h) \mid h \le j\}
  \}.
\]

\textbf{Running time.}
There are $q^{p}$ ($\le (2^{p}+p)^{p} \in \mathrm{O}(2^{p^{2}})$) candidates for the type assignment $\phi$ for $S$.
For each $\phi$, we can compute $\dptable_{\phi}[i, j]$ in time polynomial in $|V(G)|$ and $|V(H)|$ for all $(i,j)$.
We output the maximum of $\dptable_{\phi}[|V(G)|, |V(H)|]$ among all possible $\phi$.
In total, the running time is bounded by $\mathrm{O}^{*}(2^{p^{2}})$.

\ifdefined\FV
\ifdefined\APP
\section{Modifications for MCOIS parameterized by vertex cover number}
\label{sec:apx_vc}
\fi

\ifdefined\FV
\smallskip\textbf{Changes for handling induced subgraphs.}
\fi

Let us see how the algorithm in Section~\ref{sec:FPTalg-VC} for MCOS parameterized by vertex cover number can be modified for the induced-subgraph variant MCOIS\@.
The changes are obvious ones to handle induced subgraphs instead of (not necessarily induced) subgraphs.

We reject the guessed type correspondence $\phi$ if the mapping $f$ defined as $f(s_{i}) = w_{i}$ for $i \in [|S|]$
is not an ordered isomorphism between the induced subgraphs $G[S]$ and $H[\{w_{1}, \dots, w_{|S|}\}]$,
where $s_{1}, \dots, s_{|S|}$ and $w_{1}, \dots, w_{|S|}$ are the vertices defined in the preprocessing phase.

We change the definition of $\dptable_{\phi}[i,j]$ in such a way that 
it counts the maximum number of vertices in a common ordered induced subgraph.
We need no changes for the preprocessing phase since the additional test was already done for $\phi$.

In the first case of the DP table computation, we set $\dptable_{\phi}[i,j] = |S|$.
The second case needs no changes.
In the third case, we set $\mu(i,h) = 1$ if we can map $u_{i}$ to $v_{h}$, that is, we check for each $s \in S$ whether 
$u_{i}$ is adjacent to $s$ in $G$ if and only if $v_{h}$ is adjacent to the type $\phi(s)$ vertices in $H$;
otherwise, we set $\mu(i,h) = -\infty$.

\fi


\ifdefined\APP
\appendix

\section{Omitted Proofs for NP-completeness}

\ifdefined\APP
\corInterval*\label{cor:interval*}
\fi

\begin{proof}
 To show the NP-completeness for interval graphs, we construct complement graphs $\overline{G}$ and $\overline{H}$ of the threshold graphs $G$ and $H$ described in the proof of Theorem~\ref{thm:NPComp-Threshold}. 
 Both graphs $\overline{G}$ and $\overline{H}$ are threshold graphs, that is, they are interval graphs. 
 It is easy to check that the vertex orderings $\sigma$ and $\tau$ are both interval orderings. 
 The answer of the OISI on $\overline{G}$ and $\overline{H}$ is the same as $G$ and $H$ since every induced subgraph isomorphism from $H$ to $G$ is also an induced subgraph isomorphism from $\overline{H}$ to $\overline{G}$ and vice versa.

 For interval bigraphs with interval bigraph orderings, we use chain graphs $G$ and $H$ described in the proof. 
 We define the vertex orderings $\sigma = (v_n, v_{n-1}, \dots, v_1, v_{\mathrm{mid}}, v_{2n}, \dots, v_{n+1})$ and $\tau=(u_{k}, \dots, u_{1}, u_{\mathrm{mid}}, u_{2k}, \dots, u_{k+1})$. 
 The orderings are interval bigraph orderings, and the remainder of the proof is identical to the one in \cref{thm:NPComp-Threshold}. 
\end{proof}



\section{Omitted Algorithm and Proof for Signed-interval Digraphs}
\label{sec:apx_min_orderable}

\section{The recursion on FPT algorithms parameterized by pathwidth}
\label{sec:pathwidth-recursion}

\ifdefined\APP
\section{Modifications for MCOIS parameterized by vertex cover number}
\label{sec:apx_vc}
\fi

\ifdefined\FV
\smallskip\textbf{Changes for handling induced subgraphs.}
\fi

Let us see how the algorithm in Section~\ref{sec:FPTalg-VC} for MCOS parameterized by vertex cover number can be modified for the induced-subgraph variant MCOIS\@.
The changes are obvious ones to handle induced subgraphs instead of (not necessarily induced) subgraphs.

We reject the guessed type correspondence $\phi$ if the mapping $f$ defined as $f(s_{i}) = w_{i}$ for $i \in [|S|]$
is not an ordered isomorphism between the induced subgraphs $G[S]$ and $H[\{w_{1}, \dots, w_{|S|}\}]$,
where $s_{1}, \dots, s_{|S|}$ and $w_{1}, \dots, w_{|S|}$ are the vertices defined in the preprocessing phase.

We change the definition of $\dptable_{\phi}[i,j]$ in such a way that 
it counts the maximum number of vertices in a common ordered induced subgraph.
We need no changes for the preprocessing phase since the additional test was already done for $\phi$.

In the first case of the DP table computation, we set $\dptable_{\phi}[i,j] = |S|$.
The second case needs no changes.
In the third case, we set $\mu(i,h) = 1$ if we can map $u_{i}$ to $v_{h}$, that is, we check for each $s \in S$ whether 
$u_{i}$ is adjacent to $s$ in $G$ if and only if $v_{h}$ is adjacent to the type $\phi(s)$ vertices in $H$;
otherwise, we set $\mu(i,h) = -\infty$.

\fi
\end{document}